\documentclass[a4paper,10pt]{article}
\usepackage{amsmath}

\usepackage[utf8]{inputenc}
\usepackage[T1]{fontenc}
\usepackage{graphicx}
\usepackage{hyperref}
\usepackage{longtable}
\usepackage{multirow}
\usepackage{float}
\usepackage{wrapfig}
\usepackage{rotating}
\usepackage[normalem]{ulem}
\usepackage{amsmath}
\usepackage{amsthm}
\usepackage{textcomp}
\usepackage{marvosym}
\usepackage{wasysym}
\usepackage{amssymb}
\usepackage{capt-of}
\usepackage{bm}

\usepackage{comment}
\usepackage{color}
\usepackage{enumerate} 
\usepackage{myalgo}
\usepackage{bm}

\def\S{\mathbb{S}}
\def\K{S}

\def\bmx{\bm{x}}
\def\x{\bmx}
\def\bmy{\bm{y}}
\def\y{\bmy}

\def\f{\mathbf{f}}
\def\q{\mathbf{q}}

\DeclareMathOperator{\polytope}{\mathcal{C}}

\newcommand{\R}{\mathbb{R}}
\newcommand{\Z}{\mathbb{Z}}
\newcommand{\Q}{\mathbb{Q}}
\newcommand{\N}{\mathbb{N}}

\DeclareMathOperator{\bigo}{\mathcal{O}}
\DeclareMathOperator{\sgn}{sgn}
\DeclareMathOperator{\trace}{Tr}

\newcommand{\sdp}{\texttt{sdp}}
\newcommand{\sdpcon}{\texttt{sdp}}
\newcommand{\sdpfun}[3]{\texttt{sdp}(#1,#2,#3)}
\newcommand{\sdpconfun}[4]{\texttt{sdp}(#1,#2,#3,#4)}
\newcommand{\cholesky}{\texttt{cholesky}}
\newcommand{\choleskyfun}[3]{\texttt{\cholesky}(#1,#2,#3)}

\newcommand{\nsdp}{n_\text{sdp}}
\newcommand{\msdp}{m_\text{sdp}}

\newcommand{\clist}{\texttt{c\_list}}
\newcommand{\calpha}{\texttt{c\_alpha}}
\newcommand{\slist}{\texttt{s\_list}}

\newcommand{\univsostwo}{\texttt{univsos2}}
\newcommand{\intsos}{\texttt{intsos}}
\newcommand{\multivsos}{\texttt{multivsos}}
\newcommand{\intsosfun}[5]{\texttt{intsos}(#1,#2,#3,#4,#5)}
\newcommand{\absorb}{\texttt{absorb}}
\newcommand{\absorbfun}[5]{\texttt{absorb}(#1,#2,#3,#4,#5)}
\newcommand{\putinarsosfun}[6]{\texttt{Putinarsos}(#1,#2,#3,#4,#5,#6)}
\newcommand{\polyasos}{\texttt{Polyasos}}
\newcommand{\cad}{\texttt{CAD}}
\newcommand{\PP}{\texttt{RoundProject}}
\newcommand{\raglib}{\texttt{RAGLib}}
\newcommand{\putinarsos}{\texttt{Putinarsos}}

\theoremstyle{plain}
\newtheorem{theorem}{Theorem}[section]
\newtheorem{lemma}[theorem]{Lemma}
\newtheorem{proposition}[theorem]{Proposition}

\theoremstyle{definition}

\newtheorem{assumption}[theorem]{Assumption}
\newtheorem{example}{Example}
\newtheorem{remark}{Remark}

\if{
\newcommand{\qed}{\nobreak \ifvmode \relax \else
      \ifdim\lastskip<1.5em \hskip-\lastskip
      \hskip1.5em plus0em minus0.5em \fi \nobreak
      \vrule height0.75em width0.5em depth0.25em\fi}
}\fi

\newcommand{\spt}[1]{\mbox{supp}(#1)}

\definecolor{dkviolet}{rgb}{0.6,0,0.8}

\begin{document}
\author{Victor Magron$^{1,2}$ \and Mohab Safey El Din$^{2}$}
\date{\today}
\title{On Exact Polya and Putinar's Representations}

\maketitle

\footnotetext[1]{CNRS Verimag, 700 av Centrale, 38401 Saint-Martin d'Hères, France }. 
\footnotetext[2]{Sorbonne Universit\'e, CNRS, INRIA, Laboratoire d'Informatique de Paris 6, PolSys, Paris, France}



\def\xxx{{\bf (xxx)}}
\def\todo{{\bf (todo!)}}
\def\check{{\bf (check!)}}


%

\begin{abstract}
We consider the problem of finding exact sums of squares (SOS)
decompositions for certain classes of non-negative multivariate
polynomials, relying on semidefinite programming (SDP) solvers.

We start by providing a hybrid numeric-symbolic algorithm computing
exact rational SOS decompositions for polynomials lying in the
interior of the SOS cone. It computes an approximate SOS decomposition
for a perturbation of the input polynomial with an arbitrary-precision
SDP solver. An exact SOS decomposition is obtained thanks to the
perturbation terms.  We prove that bit complexity estimates on output
size and runtime are both polynomial in the degree of the input
polynomial and simply exponential in the number of variables.  Next,
we apply this algorithm to compute exact Polya and Putinar's
representations respectively for positive definite forms and positive
polynomials over basic compact semi-algebraic sets.  We also compare
the implementation of our algorithms with existing methods in computer
algebra including cylindrical algebraic decomposition and critical point
method.


\end{abstract}
\thanks{}
\section{Introduction}
\label{sec:intro}

Let $\Q$ (resp.~$\R$) be the field of rational (resp.~real) numbers
and $X = (X_1, \ldots, X_n)$ be a sequence of variables. We consider
the problem of deciding the non-negativity of $f \in \Q[X]$ either
over $\R^n$ or over a semi-algebraic set $S$ defined by some
constraints $g_1\geq 0, \ldots, g_m\geq 0$ (with $g_j \in \Q[X]$).
Further, $d$ denotes the maximum of the total degrees of these
polynomials.

This problem is known to be NP hard~\cite{blum}. The Cylindrical
Algebraic Decomposition algorithm~\cite{Collins75} allows to solve it
in time doubly exponential in $n$ (and polynomial in $d$). This
complexity result has been improved later on, through the so-called
critical point method, starting from~\cite{GV88} which culminates
with~\cite{BPR98} to establish that this decision problem can be
solved in time $((m+1)d)^{O(n)}$. These latter ones have been
developed to obtain implementations which reflect the complexity gain
(see
e.g.~\cite{BGHM1,BGHM3,SaSc03,S07,BGHSS, Safey10Issac,BGHS14,Greuet14,Greuet12Sos}) but still
within a singly exponential complexity in $n$. Besides, these
algorithms are ``root finding'' ones: they try to find a
point at which $f$ is negative over the considered domain. When $f$ is
positive, they return an empty list without a {\it certificate} that
can be checked {\it a posteriori}.

To compute certificates of non-negativity, an approach based on {\it
  sums of squares} (SOS) decompositions (and their variants) has been
popularized by Lasserre~\cite{Las01sos} and Parillo~\cite{phdParrilo}
(see also the survey~\cite{laurent2009sums} and references
therein). In a nutshell, the idea is as follows.

A polynomial $f$ is non-negative over $\R^n$ if it can be written as
an SOS $s_1^2+\cdots +s_r^2$ with $s_i \in \R[X]$ for
$1\leq i \leq r$. Also $f$ is non-negative over the semi-algebraic set
$S$ if it can be written as
$s_1^2+\cdots+s_r^2+\sum_{j=1}^m \sigma_j g_j$ where $\sigma_i$ is a
sum of squares in $\R[X]$ for $1\leq j \leq m$. It turns out that,
thanks to the ``Gram matrix method'' (see
e.g.~\cite{Las01sos,phdParrilo}), computing such decompositions can be
reduced to solving Linear Matrix Inequalities (LMI). This boils down
to considering a semidefinite programming (SDP) problem.

For instance, on input $f \in \Q[X]$ of even degree $d = 2k$, the
decomposition $f=s_1^2+\cdots+s_r^2$ is a by-product of a
decomposition of the form $f = v_k^T L^T D L v_k$ where $v_k$ is the
vector of all monomials of degree $\leq k$ in $\Q[X]$, $L$ is a lower
triangular matrix with non-negative real entries on the diagonal and
$D$ is a diagonal matrix with non-negative real entries. The matrices
$L$ and $D$ are obtained after computing a symmetric matrix $G$ (the
Gram matrix), semidefinite positive, such that $f = v_k^T G v_k$. Such
a matrix $G$ is found using solvers for LMIs. Such inequalities can be
solved symbolically (see~\cite{Simone16}), but the degrees of the
algebraic extensions needed to encode exactly the solutions are
prohibitive on large examples \cite{NRS10}.
Besides, there exist fast numerical solvers for solving LMIs
implemented in double precision,
e.g.~SeDuMi~\cite{Sturm98usingsedumi},~SDPA~\cite{Yamashita10SDPA} as
well as arbitrary-precision solvers, e.g.~SDPA-GMP~\cite{Nakata10GMP},
successfully applied in many contexts, including bounds for kissing
numbers~\cite{Bachoc06newupper} or computation of (real) radical
ideals~\cite{Las13ideal}.

But using uniquely numerical solvers yields ``approximate''
non-negativity certificates. On our example, the matrices $L$ and $D$
(and consequently the polynomials $s_1, \ldots, s_r$) are not known
exactly.

This raises topical questions. The first one is how to let interact
symbolic computation with these numerical solvers to get {\it exact}
certificates? Since not all positive polynomials are SOS, what to do when SOS certificates do not
exist? Also, given inputs with
rational coefficients, can we obtain certificates with rational
coefficients?

%

For these questions, we inherit from previous contributions in the
univariate case~\cite{Chevillard11,univsos} as well as in the
multivariate case~\cite{PaPe08,KLYZ08}. Diophantine aspects are
considered in \cite{Safey10Siam, GSZ13}. 
In the univariate (un)-constrained case, the algorithm
from~\cite{Chevillard11} computes an exact weighted SOS decomposition
for a given positive polynomial $f \in \Q[X]$. The algorithm considers
a perturbation of $f$, performs (complex) root isolation to get an
approximate SOS decomposition of $f$. When the isolation is precise
enough, the algorithm relies the perturbation terms to recover an
exact rational decomposition.
In the multivariate unconstrained case, Parillo and Peyrl designed a rounding-projection algorithm in~\cite{PaPe08} to compute a weighted rational SOS decompositon of a given polynomial $f$ in the interior of the SOS cone. The algorithm computes an approximate Gram matrix of $f$, and rounds it to a rational matrix.  With sufficient precision digits, the algorithm performs an orthogonal projection to recover an exact Gram matrix of $f$. The SOS decomposition is then obtained with an exact $L D L^T$ procedure. 
{This approach was significantly extended in~\cite{KLYZ08} to handle rational functions.}


\textit{Main contributions.} This work provides an algorithmic
framework to handle (un)-constrained polynomial problems
with exact rational weighted SOS decompositions. The first contribution, given
in Section~\ref{sec:intsos}, is a hybrid numeric-symbolic algorithm,
called $\intsos$, providing rational SOS decompositions for
polynomials lying in the interior of the SOS cone. 
As for the algorithm from~\cite{Chevillard11}, the main idea is to
perturbate the input polynomial, then to obtain an approximate Gram matrix
of the perturbation by solving an SDP problem, and to recover an exact
decomposition with the perturbation terms.

In Section~\ref{sec:polya}, we rely on $\intsos$ to compute decompositions of positive definite forms into SOS of rational functions, based on
Polya's representations, yielding a second algorithm, called
$\polyasos$.
In Section~\ref{sec:putinar}, we rely on $\intsos$ to
compute weighted SOS decompositions for polynomials positive over
compact semi-algebraic sets, yielding a third algorithm, called
$\putinarsos$.

When the input is an $n$-variate polynomial of degree $d$ with integer
coefficients of maximum bit size $\tau$, we prove in
Section~\ref{sec:intsos} that Algorithm~$\intsos$ runs in boolean time
$\tau^2 d^{\bigo{(n)}}$ and outputs SOS polynomials of bit size
bounded by $\tau d^{\bigo{(n)}}$.  This also yields bit complexity
analysis for Algorithm~$\polyasos$ (see Section~\ref{sec:polya}) and
Algorithm~$\putinarsos$ (see Section~\ref{sec:putinar}). To the best
of our knowledge, these are the first complexity estimates for the
output of algorithms providing exact multivariate SOS decompositions.

The three algorithms are implemented within a Maple library, called
$\multivsos$. In Section~\ref{sec:benchs}, we provide numerical
benchmarks to evaluate the performance of $\multivsos$ against
existing methods based on CAD or critical point methods.

\paragraph*{Acknowledgments.} M.~Safey El Din is supported by the
ANR-17-CE40-0009 GALOP project and the GAMMA project funded by
PGMO/FMJH. V.~Magron is supported by the LabEx PERSYVAL-Lab (ANR-11-LABX-0025-01) funded by the French program ``Investissement d'avenir'' and by the European Research Council (ERC) ``STATOR'' Grant Agreement nr. 306595.


%


\vspace{-0.4cm}
\section{Preliminaries}
\label{sec:prelim}

Let $\Z$ be the set of integers. 
For
$\alpha = (\alpha_1,\dots,\alpha_n) \in \N^n$, one has
$|\alpha| := \alpha_1 + \dots + \alpha_n$ and
$X^\alpha := X_1^{\alpha_1} \dots X_n^{\alpha_n}$.
For all $k \in \N$, we let
$\N^{n}_k := \{ \alpha \in \N^{n} : |\alpha| \leq k \}$, whose
cardinality is the binomial $\binom{n+k}{k}$.  A polynomial
$f \in \R[X]$ of degree $d = 2k$ is written as
$ f \,=\,\sum_{|\alpha|\leq d} \, f_{\alpha} \, X^\alpha $ and we
identify $f$ with its vector of coefficients $\f=(f_{\alpha})$ in the
basis $(X^\alpha)$, $\alpha \in\N_d^n$. 
Let ${\Sigma}[X]$ be the convex cone of sums of squares in $\R[X]$ and
$\mathring{\Sigma}[X]$ be the interior of $\Sigma[X]$. We note $\Sigma_\Z(X) := \Z[X] \cap \Sigma[X]$ and $\mathring{\Sigma}_\Z[X]$ its interior. For instance, the polynomial $f = 4 X_1^4 + 4 X_1^3 X_2 - 7 X_1^2 X_2^2 - 2 X_1 X_2^3 + 10 X_2^4 = (2 X_1 X_2 + X_2^2)^2 + (2 X_1^2 + X_1 X_2 - 3 X_2^2)^2$ belongs to $\Sigma_\Z(X)$.

We rely on the bit complexity model for complexity estimates. The bit size of an integer $b$ is denoted by $\tau(b) := \log_2 (|b|) + 1$ with $\tau(0) := 1$. For $f = \sum_{|\alpha| \leq d} f_\alpha X^\alpha \in \Z[X]$ of degree $d$, we note $\|f\|_\infty := \max_{|\alpha| \leq d} |f_\alpha|$ and $\tau(f) := \tau(\|f\|_\infty)$ with slight abuse of notation. Given $b \in \Z$ and $c \in \Z \backslash \{0\}$ with gcd$(b,c) = 1$, we define $\tau(b/c) := \max \{\tau(b), \tau(c)\}$. For two mappings $g,h : \N^l \to \R$, we use the notation ``$g(v) = \bigo{(h(v))}$'' to state the existence of $b \in \N$ such that $g(v) \leq b h(v)$, for all $v \in \N^l$.

The {\em Newton polytope}
or {\em cage} $\polytope{(f)}$ is the convex hull of the vectors of exponents
of monomials that occur in $f \in \R[X]$.  For the above example, $\polytope{(f)} = \{(4,0),(3,1),(2,2),(1,3),(0,4)\}$. For a symmetric real matrix $G$, we note $G \succeq 0$ (resp.~$G \succ 0$) when $G$ has only non-negative (resp.~positive) eigenvalues and we say that $G$ is {\em positive semidefinite} (SDP) (resp.~{\em positive definite}).

\if{
Given a real sequence
$\y =(y_{\alpha})_{\alpha \in \N^n}$, the linear functional
$\ell_\y : \R[X] \to \R$ is defined by
$\ell_\y(f) := \sum_{|\alpha| \leq d} f_{\alpha} y_{\alpha}$, for all
$f \in \R[X]$. 
We associate to $\y$ the {\it moment matrix} $M_k(\y)$, that is the real
symmetric matrix with rows (resp.~columns) indexed by $\N_k^n$ and the
following entrywise definition:
$(M_k(\y))_{\beta,\gamma} := \ell_\y(X^{\beta + \gamma})$ forall
$\beta, \gamma \in \N_k^n$.

Given $g \in \R[X]$, we associate to $\y$ the {\it localizing matrix},
that is the real symmetric matrix $M_k(g \, \y)$ with rows (resp.~
columns) indexed by $\N_k^{n}$ and the following entrywise definition:
$(M_k(g \, \y))_{\beta, \gamma} := \ell_\y(g \, X^{\beta + \gamma}) $
forall $\beta, \gamma \in \N_k^n $.
\begin{equation}
\label{eq:primalsdp}
\begin{aligned}
\inf\limits_{\y} \quad & \sum_{|\alpha| \leq d} f_\alpha y_{\alpha}   \quad 
\text{s.t.} \quad & M_k(\y) \succeq 0 \,, \quad y_0 = 1  \,. \\
\end{aligned}
\end{equation}
Writing $M_k(\y) = \sum_{|\alpha| \leq d} y_\alpha B_\alpha$, the dual of SDP~\eqref{eq:primalsdp} is also an SDP:
}\fi
With $f \in \R[X]$ of degree $d = 2k$, we consider the SDP program:
\begin{equation}
\label{eq:dualsdp}
\begin{aligned}
\inf\limits_{G \succeq 0} \quad & \trace{(G \, B_0)}  \quad 
\text{s.t.} \quad & \trace{(G \, B_\gamma)} = f_\gamma \,, \quad \forall \gamma \in \N_d^n \,, \\
\end{aligned}
\end{equation}
where $B_\gamma$  has rows (resp.~
columns) indexed by $\N_k^{n}$ with $(\alpha, \beta)$ entry equal to 1 if $\alpha + \beta = \gamma$ and 0 otherwise.
\begin{theorem}\cite[Theorem~3.2]{Las01sos}
\label{th:lasunconstrained}
Let $f \in \R[X]$ of degree $d = 2k$ and global minimum
$f^\star := \inf_{\bmx \in \R^n} f(\bmx)$.  
Assume that SDP~\eqref{eq:dualsdp} has a feasible solution
$G^\star = \sum_{i=1}^r \lambda_i \q_i \, \q_i^T $, with the $\q_i$
being the eigenvectors of $G^\star$ corresponding to the non-negative
eigenvalues $\lambda_i$, for all $i=1,\dots,r$.  Then
$f - f^\star = \sum_{i=1}^r \lambda_i q_i^2$.
\end{theorem}
For the sake of efficiency, one reduces the size of matrix $G$
indexing its rows and columns by half of $\polytope{(f)}$:
\vspace*{-0.1cm}
\begin{theorem}\cite[Theorem~1]{Reznick78}
\label{th:np}
Let $f \in \Sigma[X]$ with $f = \sum_{i=1}^r s_i^2$ and $P := \polytope{(f)}$.  Then for all $i=1,\dots,r$, $\polytope{(s_i)} \subseteq P/2$.
\end{theorem}
\vspace*{-0.1cm}

Given $f \in \R[X]$, Theorem~\ref{th:lasunconstrained} states that one can theoretically certify that $f$ lies in
$\Sigma[X]$ by solving SDP~\eqref{eq:dualsdp}.  However,
available SDP solvers are typically implemented in finite-precision
and require the existence of a strictly feasible solution $G \succ 0$
to converge. This is equivalent for $f$ to lie in
$\mathring{\Sigma}[X]$ as stated in~\cite[Proposition~5.5]{Choi95}:
\vspace*{-0.1cm}
\begin{theorem}
\label{th:intsospdGram}
Let $f \in \Z[X]$ with $P := \polytope{(f)}$ and $v_k$ be the
vector of all monomials in $P/2$. Then $f \in \mathring{\Sigma}[X]$
if and only if there exists a positive definite matrix $G$ such that
$f = v_k^T G v_k$.
\end{theorem}
\vspace{-0.3cm}

%
\section{Exact SOS representations}
\label{sec:intsos}
The aim of this section is to state and analyze a hybrid numeric-symbolic algorithm, called $\intsos$, computing weighted SOS decompositions of polynomials in $\mathring{\Sigma}_\Z[X]$. This algorithm relies on perturbations of such polynomials. 
\begin{proposition}
\label{th:boundeps}
Let $f \in \mathring{\Sigma}_\Z[X]$ of degree $d = 2k$, with
$\tau = \tau(f)$ and $P=\polytope(f)$. Then, there exists
$N\in \N-\{0\}$ such that for $\varepsilon := \frac{1}{2^N}$, 
$f - \varepsilon \sum_{\alpha \in P/2} X^{2 \alpha}
\in \mathring{\Sigma}[X]$.
Moreover, $N = \tau(\varepsilon) \leq \tau d^{\bigo{(n)}}$.
\end{proposition}
\begin{proof}
  Let $v_k$ be the vector of all monomials $X^\alpha$ in $P/2$. Note
  that each monomial in $v_k$ has degree $\leq k$ and that
  $v_k^T v_k = \sum_{\alpha \in P/2} X^{2 \alpha}$. Since
  $f \in \mathring{\Sigma}[X]$, there exists by Theorem~\ref{th:intsospdGram} a 
  matrix $G \succ 0$ such that $f =v_k^T G v_k$, with positive smallest eigenvalue
  $\lambda$. Let us define
  $N := \lceil \log_2 \frac{1}{\lambda} \rceil + 1$, i.e.~the smallest
  integer such that
  $\varepsilon = \frac{1}{2^N} \leq \frac{\lambda}{2}$. Then,
  $\lambda > \varepsilon$ and the matrix $G - \varepsilon I$ has only
  positive eigenvalues. Hence, one has
\[f_\varepsilon := f - \varepsilon \sum_{\alpha \in P/2} X^{2 \alpha}
  = v_k^T G v_k - \varepsilon v_k^T I v_k = v_k^T (G - \varepsilon I)
  v_k \,,\]
  yielding $f_\varepsilon \in \mathring{\Sigma}[X]$.  
  
  For the second
  claim, let us consider the set
  $A := \{e \in \R : \forall \bmx \in \R^n, f(\bmx) - e \sum_{\alpha
    \in P/2} \bmx^{2 \alpha} \geq 0\}$.
  Using \cite[Thm 14.16]{BPR06}, $A$ is defined by univariate
  polynomials of degree in $d^{\bigo{(n)}}$ with coefficients of bit
  size bounded by $\tau d^{\bigo{(n)}}$. Hence the bit size of the
  mimimum absolute value of their non-zero real roots is below bounded
  by $\tau d^{\bigo{(n)}}$.
\end{proof}
The following can be found in~\cite[Lemma~2.1]{Bai89}
and~\cite[Theorem~3.2]{Bai89}.
\begin{proposition}
\label{th:boundchol}
Let $\tilde{G} \succ 0$ be a matrix with
rational entries indexed on $\N^n_r$. Let $L$ be the factor of $\tilde{G}$
computed using Cholesky's decomposition with finite precision
$\delta_c$. Then $L L^T = \tilde{G} + E$ where
\begin{align}
\label{eq:chol1}
|E_{\alpha, \beta}| \leq {(r+1)2^{-\delta_c}} |\tilde{G}_{\alpha,\alpha} \, \tilde{G}_{\beta,\beta}|^{\frac{1}{2}} / (1 - (r+1)2^{-\delta_c})  \,.
\end{align}
In addition, if the smallest eigenvalue 
$\tilde{\lambda}$ of $\tilde{G}$ satisfies the  inequality
\begin{align}
\label{eq:chol2}
2^{-\delta_c} < \tilde{\lambda} / (r^2 + r + (r-1) \tilde{\lambda}) \,,
\end{align}
Cholesky's decomposition returns a rational nonsingular factor $L$.
\end{proposition}

\subsection{Algorithm~\texttt{intsos}}
\label{sec:algo}
We present our algorithm $\intsos$ computing exact weighted rational SOS
decompositions for polynomials in $\mathring{\Sigma}_\Z[X]$. 

{
\begin{algorithm}[t]
\caption{$\intsos$
}
\label{alg:intsos}
\begin{algorithmic}[1]
\Require $f \in \Z[X]$, positive $\varepsilon \in \Q$, precision parameters $\delta, R \in \N$ for the SDP solver 
,  precision $\delta_c \in \N$ for the Cholesky's decomposition
\Ensure list $\clist$ of numbers in $\Q$ and list $\slist$ of polynomials in $\Q[X]$
\State $P := \polytope{(f)}$ \label{line:np}
\State $t := \sum_{\alpha \in P/2} X^{2 \alpha}$, $f_\varepsilon \gets f - \varepsilon t$
\While {$f_\varepsilon \notin \mathring{\Sigma}[X]$} \label{line:epsi}
 $\varepsilon \gets \frac{\varepsilon}{2}$, $f_\varepsilon \gets f - \varepsilon t$
\EndWhile \label{line:epsf}
\State ok := false
\While {not ok} \label{line:deltai}
\State $(\tilde{G}, \tilde{\lambda}) \gets \sdpfun{f_\varepsilon}{ \delta}{R}$ \label{line:sdp}
\State $(s_1,\dots,s_r) \gets \choleskyfun{\tilde{G}}{\tilde{\lambda}}{\delta_c}$ \label{line:chol} \Comment{$f_\varepsilon \simeq \sum_{i=1}^r  s_i^2$}
\State $u \gets f_\varepsilon - \sum_{i=1}^r  s_i^2$
\State $\clist \gets [1,\dots,1] $, $\slist \gets [s_1,\dots,s_r]$
\For {$\alpha \in P/2$}  $\varepsilon_{\alpha} := \varepsilon$
\EndFor
\State $\clist, \slist,(\varepsilon_\alpha) \gets \absorbfun{u}{P}{(\varepsilon_\alpha)}{\clist}{\slist}$ \label{line:absorb}
\If {$\min_{\alpha \in P/2} \{ \varepsilon_\alpha \} \geq 0$} ok := true \label{line:sosok}
\Else $\ \delta \gets 2 \delta$, $R \gets 2 R$, $\delta_c \gets 2 \delta_c$
\EndIf
\EndWhile \label{line:deltaf}
\For {$\alpha \in P/2$}  
\State $\clist \gets  \clist \cup \{ \varepsilon_\alpha \}$, $\slist \gets  \slist \cup \{ X^\alpha \}$
\EndFor
\State \Return $\clist$, $\slist$
\end{algorithmic}
\end{algorithm}
}

{
\begin{algorithm}
\caption{$\absorb$
}
\label{alg:absorb}
\begin{algorithmic}[1]
\Require $u \in \Q[X]$, multi-index set $P$, lists $(\varepsilon_\alpha)$ and $\clist$ of numbers in $\Q$, list $\slist$ of polynomials in $\Q[X]$
\Ensure lists $(\varepsilon_\alpha)$ and $\clist$ of numbers in $\Q$, list $\slist$ of polynomials in $\Q[X]$
\For {$\gamma \in \spt{u}$} \label{line:absi}
\If {$\gamma \in (2\N)^n$} $\alpha := \frac{\gamma}{2}$, $\varepsilon_\alpha := \varepsilon_\alpha + u_\gamma$ \label{line:intsoseven}
\Else \State Find $\alpha$, $\beta \in P/2$ such that $\gamma = \alpha + \beta $ \label{line:intsosodd}
\State $\varepsilon_\alpha := \varepsilon_\alpha - \frac{|u_\gamma|}{2}$, $\varepsilon_\beta := \varepsilon_\beta - \frac{|u_\gamma|}{2}$
\State $\clist \gets  \clist \cup \{ \frac{|u_\gamma|}{2} \}$
\State $\slist \gets  \slist \cup \{ X^\alpha + \sgn{(u_\gamma)} X^\beta \}$
\EndIf
\EndFor \label{line:absf}
\end{algorithmic}
\end{algorithm}
}

Given $f \in \Z[X]$ of degree $d = 2k$, one first 
computes its Newton polytope $P := \polytope{(f)}$ (see
line~\lineref{line:np}) using standard algorithms such as
quickhull~\cite{Barber96}. The
loop going from line~\lineref{line:epsi} to line~\lineref{line:epsf} finds a positive
 $\varepsilon \in \Q$ such that the perturbed polynomial
$f_\varepsilon := f - \varepsilon \sum_{\alpha \in P/2} X^{2 \alpha}$
is also in $\mathring{\Sigma}[X]$. This is done thanks to 
an oracle based on SDP or computer algebra procedures (e.g.~CAD or critical points). If $f \in \mathring{\Sigma}_\Z[X]$, the existence of $\varepsilon$ is
ensured as in the proof of Theorem~\ref{th:boundeps} if $A := \{e \in \R : \forall \bmx \in \R^n, f(\bmx) - e \sum_{\alpha
    \in P/2} \bmx^{2 \alpha} \geq 0\}$ is non empty.

Next, we enter in the loop starting from line~\lineref{line:deltai}.  Given $f_\varepsilon \in \Z[X]$, positive integers $\delta$
and $R$, the $\sdp$ function calls an SDP solver and tries to
compute a rational approximation $\tilde{G}$ of the Gram matrix
associated to $f_\varepsilon$ together with a rational approximation
$\tilde{\lambda}$ of its smallest eigenvalue. 
In practice, we use an arbitrary-precision SDP solver implemented with an interior-point method. However, in order to analyse the
complexity of the procedure (see Remark~\ref{rk:ellipsoidIP}), we assume that $\sdp$ relies on
the ellipsoid algorithm~\cite{GroetschelLovaszSchrijver93}.
\begin{remark}
\label{rk:ellipsoidIP}
In~\cite{deKlerkSDP}, the authors analyze the complexity of the short step, primal interior point method, used in SDP solvers. 
Within fixed accuracy, they obtain a polynomial complexity, as for the ellipsoid method, but the exact value of the exponents is not provided.
\end{remark}
SDP problems are solved with this latter algorithm in polynomial-time within a
given accuracy $\delta$ and a radius bound $R$ on the Frobenius norm of $\tilde{G}$.
The first
step consists of solving SDP~\eqref{eq:dualsdp} by computing an
approximate Gram matrix $\tilde{G} \succeq 2^{- \delta} I $ such that
$|\trace{( \tilde{G} B_\gamma)} - (f_\varepsilon)_{\gamma}| = | \sum_{\alpha+\beta = \gamma} \tilde{G}_{\alpha,\beta} -
(f_\varepsilon)_{\gamma}| \leq
2^{-\delta}$ 
and $\sqrt{\trace{(\tilde{G}^2)}} \leq R$.
We pick large enough $\delta$ and $R$ to obtain $\tilde{G} \succ 0$ and $\tilde{\lambda} > 0$ when $f_\varepsilon \in \mathring{\Sigma}[X]$.

The $\cholesky$ function computes the approximate
Cholesky's decomposition $L L^T$ of $\tilde{G}$ with precision
$\delta_c$. In order to guarantee that $L$ will be a rational
nonsingular matrix, a preliminary step consists of verifying that the
inequality from~\eqref{eq:chol2} holds, which happens when $\delta_c$
is large enough. Otherwise, $\cholesky$  selects the
smallest $\delta_c$ such as~\eqref{eq:chol2} holds.  Let $v_k$ be the
vector of all monomials $X^\alpha$ belonging to $P/2$ with size $r$. The output is a list of
rational polynomials $[s_1,\dots,s_r]$ such that for all
$i=1,\dots,r$, $s_i$ is the inner product of the $i$-th row of $L$
by $v_k$. By Theorem~\ref{th:lasunconstrained}, one would have $f_\varepsilon = \sum_{i=1}^r s_i^2$ with $s_i \in \R[X]$ after using exact SDP and Cholesky's decomposition. Here, we have to consider  the remainder $u = f - \varepsilon \sum_{\alpha \in P/2} X^{2 \alpha} - \sum_{i=1}^r s_i^2$, with $s_i \in \Q[X]$. 

After these numeric steps, the algorithm starts to perform symbolic computation with the $\absorb$ subroutine at line~\lineref{line:absorb}.
The loop from $\absorb$ is designed to obtain an exact weigthed SOS decomposition of $\varepsilon t + u = \varepsilon \sum_{\alpha \in P/2} X^{2 \alpha} + \sum_{\gamma} u_\gamma X^\gamma$, yielding in turn an exact decomposition of $f$. Each term $u_\gamma X^\gamma$ can be written either $u_\gamma X^{2 \alpha}$ or $u_\gamma X^{\alpha + \beta}$, for $\alpha, \beta \in P/2$. In the former case (line~\lineref{line:intsoseven}), one has $\varepsilon X^{2 \alpha} + u_\gamma X^{2 \alpha} = (\varepsilon + u_\gamma) X^{2 \alpha}$.
In the latter case (line~\lineref{line:intsosodd}), one has 
\[
\varepsilon (X^{2 \alpha} +  X^{2 \beta}) + u_\gamma X^{\alpha + \beta} = |u_{\gamma}|/2 (X^{\alpha} + \sgn{(u_\gamma)} X^{\beta})^2 + (\varepsilon - |u_{\gamma}|/2) (X^{2 \alpha} + X^{2 \beta}) \,.\]
If the positivity test of line~\lineref{line:sosok} fails, then the coefficients of $u$ are too large and one cannot ensure that $\varepsilon t + u$ is SOS. So we repeat the same procedure after increasing the precision of the SDP solver and Cholesky's decomposition.

In prior work~\cite{univsos}, the authors and Schweighofer formalized and analyzed an algorithm called $\univsostwo$, initially provided in~\cite{Chevillard11}. Given a univariate polynomial $f > 0$ of degree $d = 2k$, this algorithm computes weighted SOS decompositions of $f$. With $t := \sum_{i=0}^k X^{2 i}$, the first numeric step of $\univsostwo$ is to find $\varepsilon$ such that the perturbed polynomial $f_\varepsilon := f - \varepsilon t > 0$ and to compute its complex roots, yielding an approximate SOS decomposition $s_1^2 + s_2^2$. The second symbolic step is very similar to the loop from line~\lineref{line:absi} to line~\lineref{line:absf} in $\intsos$: one considers the remainder polynomial $u := f_\varepsilon - s_1^2 - s_2^2$ and tries to computes an exact SOS decomposition of $\varepsilon t + u$. This succeeds for large enough precision of the root isolation procedure.
Therefore, $\intsos$ can be seen as an extension of $\univsostwo$ in the multivariate case by replacing the numeric step of root isolation by SDP and keeping the same symbolic step.

\begin{example}
\label{ex:intsos}
We apply Algorithm~$\intsos$ on $f = 4 X_1^4 + 4 X_1^3 X_2 - 7 X_1^2 X_2^2 - 2 X_1 X_2^3 + 10 X_2^4$, with $\varepsilon = 1$, $\delta = R = 60$ and $\delta_c = 10$. Then $P/2 := \polytope{(f)}/2 = \{(2,0),(1,1),(0,2)\}$ (line~\lineref{line:np}). The loop from line~\lineref{line:epsi} to line~\lineref{line:epsf} ends and we get $f - \varepsilon t = f - (X_1^4 + X_1^2 X_2^2 + X_2^2) \in  \mathring{\Sigma}[X]$. The $\sdp$ (line~\lineref{line:sdp}) and $\cholesky$ (line~\lineref{line:chol}) procedures yield $s_1 = 2 X_1^2+ X_1 X_2- \frac{8}{3} X_2^2$, $s_2 = \frac{4}{3} X_1 X_2+ \frac{3}{2} X_2^2$ and $s_3 =  \frac{2}{7} X_2^2$. The remainder polynomial is $u = f- \varepsilon t - s_1^2 - s_2^2 - s_3^2 = -  X_1^4- \frac{1}{9} X_1^2 X_2^2- \frac{2}{3} X_1 X_2^3- \frac{781}{1764} X_2^4$.

At the end of the loop from line~\lineref{line:absi} to line~\lineref{line:absf}, we obtain $\varepsilon_{(2,0)} = (\varepsilon -  X_1^4 = 0$, which is the coefficient of $X_1^4$ in $\varepsilon t + u$.  
Then, $\varepsilon (X_1^2 X_2^2 + X_2^4) - \frac{2}{3} X_1 X_2^3 = \frac{1}{3} (X_1 X_2 - X_2^2)^2 + (\varepsilon - \frac{1}{3}) (X_1^2 X_2^2 + X_2^4)$. In the polynomial $\varepsilon t + u$, the coefficient of $X_1^2 X_2^2$ is $\varepsilon_{(1,1)} = \varepsilon - \frac{1}{3} - \frac{1}{9} = \frac{5}{9}$ and the coefficient of $X_4^4$ is $\varepsilon_{(0,2)} = \varepsilon - \frac{1}{3} - \frac{781}{1764} = \frac{395}{1764}$. 

Eventually, we obtain the weighted rational SOS decomposition: $4 X_1^4 + 4 X_1^3 X_2 - 7 X_1^2 X_2^2 - 2 X_1 X_2^3 + 10 X_2^4 = \frac{1}{3} (X_1 X_2-X_2^2)^2 + \frac{5}{9} (X_1 X_2)^2 + \frac{395}{1764} X_2^4 + (2 X_1^2+ X_1 X_2- \frac{8}{3} X_2^2)^2 + (\frac{4}{3} X_1 X_2+ \frac{3}{2}  X_2^2)^2+(\frac{2}{7} X_2^2)^2)$.
\end{example}
\subsection{Correctness and bit size of the output}
\label{sec:bitsize}
%
%
Let $f \in \mathring{\Sigma}_\Z[X]$ of degree $d = 2k$,
$\tau = \tau(f)$ and $P = \polytope(f)$.
\begin{proposition}
\label{th:boundR}
Let $G$ be a positive definite Gram matrix associated to $f$ and
$0 < \epsilon \in \Q$ be such that
$f_\varepsilon = f - \varepsilon \sum_{\alpha \in P/2} X^{2
  \alpha}\in \mathring{\Sigma}[X]$.
Then, there exist positive integers $\delta$, $R$ such that $G - \varepsilon I$ is a Gram matrix associated to $f_\varepsilon$, satisfies
$G - \varepsilon I \succeq 2^{- \delta} I$ and
$\sqrt{\trace{({G-\varepsilon I}^2)}} \leq R$. Also, the maximal
bit sizes of $\delta$ and $R$ are upper bounded by
$\tau d^{\bigo{(n)}}$.
\end{proposition}
\begin{proof}
  Let $\lambda$ be the smallest eigenvalue of $G$. By
  Proposition~\ref{th:boundeps}, $G \succeq \varepsilon I$ for
  $\varepsilon = \frac{1}{2^N} \leq \frac{\lambda}{2}$. With
  $\delta = N+1$, 
  $2^{-\delta} = \frac{1}{2^{N+1}} \leq \frac{\lambda}{4} <
  \frac{\lambda}{2}$,
  yielding
  $G - \varepsilon \succeq \frac{\lambda}{2} I \succeq 2^{- \delta}I$.
  As $N \leq \tau d^{\bigo{(n)}}$, one has
  $\delta \leq \tau d^{\bigo{(n)}}$.

  As in the proof of Proposition~\ref{th:boundeps}, we consider the
  largest eigenvalue $\lambda'$ of the Gram matrix $G$ of $f$ and
  prove that the set
  $A':= \{e' \in \R : \forall \x \in \R^n, - f(\x) + e' \sum_{\alpha
    \in P/2} \x^{2 \alpha} \geq 0\}$
  is not empty. We use again \cite[Thm 14.16]{BPR06} to prove that
  $A'$ contains an interval $]0, \frac{1}{2^N}[$ with
  $N \leq \tau d^{\bigo{(n)}}$. This allows in turn to obtain a
  rational upper bound $\varepsilon'$ of $\lambda'$ with bit size
  $\tau d^{\bigo{(n)}})$. The size of $G$ is bounded by
  $\binom{n+k}{n}$, thus the trace of $G^2$ is less than 
  $\binom{n+k}{n} \varepsilon'^2$. Using that for all $k \geq 2$, 
\[\binom{n+k}{n} = \frac{(n+k) \cdots (k+1)}{n!} =
(1+\frac{k}{n})(1+\frac{k}{n-1})\cdots(1+k)     \leq k^{n-1} (1+k) \leq 2 k^n \leq d^n \,,
\] 
one has
  $\sqrt{\trace{({G-\varepsilon I})^2}} \leq d^{\frac{n}{2}}
  \varepsilon' = \tau d^{\bigo{(n)}}$.
\end{proof}
\begin{proposition}
\label{th:bitmultivsos}
Let $f$ be as above. When applying Algorithm~$\intsos$ to $f$, the
procedure always terminates and outputs a weighted rational SOS decompositon of
$f$. The maximum bit size of the coefficients involved in this SOS
decomposition is upper bounded by $\tau d^{\bigo{(n)}}$.
\end{proposition}
\begin{proof}
  Let us first consider the loop of Algorithm~$\intsos$ defined from
  line~\lineref{line:epsi} to line~\lineref{line:epsf}.  From
  Proposition~\ref{th:boundeps}, this loop terminates when
  $f_\varepsilon \in \mathring{\Sigma}[X]$ for
  $\varepsilon = \frac{1}{2^N}$ and $N \leq \tau d^{\bigo{(n)}}$.
 
  When calling the $\sdp$ function at line~\lineref{line:sdp} to solve
  SDP~\eqref{eq:dualsdp} with precision parameters $\delta$ and $R$,
  we compute an approximate Gram matrix $\tilde{G}$ of $f_\varepsilon$
  such that $\tilde{G} \succeq 2^{\delta} I $ and
  $\trace{(\tilde{G}^2)} \leq R^2$.  From
  Proposition~\ref{th:boundR}, this procedure succeeds for large
  enough values of $\delta$ and $R$ of bitisze upper bounded by
  $\tau d^{\bigo{(n)}}$. In this case, we obtain a positive rational
  approximation $\tilde{\lambda} \geq 2^{-\delta}$ of the smallest
  eigenvalue of $\tilde{G}$.

  Then the Cholesky's decomposition of $\tilde{G}$ is computed when
  calling the $\cholesky$ function at line~\lineref{line:chol}.  The
  decomposition is guaranteed to succeed by selecting a large enough
  $\delta_c$ such that~\eqref{eq:chol2} holds. Let $r$ be the size of
  $\tilde{G}$ and $\delta_c$ be the smallest integer such that
  $2^{-\delta_c} < \frac{2^{-\delta}}{r^2 + r + (r-1) 2^{-\delta}}$.
  Since the function $x \mapsto \frac{x}{r^2 + r + (r-1) x}$ is
  increasing on $[0,\infty)$ and
  $\tilde{\lambda} \geq 2^{-\delta}$,~\eqref{eq:chol2} holds. We
  obtain an approximate weighted SOS decomposition
  $\sum_{i=1}^r s_i^2$ of $f_\varepsilon$ with rational coefficients.

  Let us now consider the remainder polynomial
  $u = f_\varepsilon - \sum_{i=1}^r s_i^2$. The second loop of
  Algorithm~$\intsos$ defined from line~\lineref{line:deltai} to
  line~\lineref{line:deltaf} terminates when for all $\alpha \in P/2$,
  $\varepsilon_\alpha \geq 0$.  This condition is fulfilled when for
  all $\alpha \in P/2$,
  $\varepsilon- \sum_{\beta \in P/2} |u_{\alpha + \beta}|/2 + u_\alpha
  \geq 0$.
  This latter condition holds when for all
  $\gamma \in \spt{u}$, $|u_\gamma| \leq \frac{\varepsilon}{r}$.

  Next, we show that this happens when the precisions $\delta$ of 
  $\sdp$ and $\delta_c$ of $\cholesky$ are both large enough.  From the definition
  of $u$, one has for all $\gamma \in \spt{u}$,
  $u_\gamma = f_\gamma - \varepsilon_{\gamma} - (\sum_{i=1}^r
  s_i^2)_\gamma$,
  where $\varepsilon_\gamma = \varepsilon$ when $\gamma \in (2\N)^n$
  and $\varepsilon_\gamma = 0$ otherwise.  The positive definite
  matrix $\tilde{G}$ computed by the SDP solver is an approximation of
  an exact Gram matrix of $f_\varepsilon$.  At precision $\delta$, one
  has for all $\gamma \in \spt{f}$,
  $\tilde{G} \succeq 2^{- \delta} I $ such that
\[
|f_\gamma - \varepsilon_\gamma - \trace{( \tilde{G} B_\gamma)} | = |f_\gamma - \varepsilon_\gamma - \sum_{\alpha + \beta = \gamma}
  \tilde{G}_{\alpha, \beta}|
  \leq 2^{-\delta}
 \,. \]
%

  In addition, it follows from~\eqref{eq:chol1} that the approximated
  Cholesky decomposition $L L^T$ of $\tilde{G}$ performed at precision
  $\delta$ satisfies $L L^T = \tilde{G} + E$ with
  $|E_{\alpha,\beta} | \leq \frac{(r+1)2^{-\delta_c}}{1 -
    (r+1)2^{-\delta_c}} |\tilde{G}_{\alpha,\alpha} \,
  \tilde{G}_{\beta,\beta}|^{\frac{1}{2}}$,
  for all $\alpha,\beta \in P/2$.
Moreover, by using Cauchy-Schwartz inequality, one has
\[\sum_{\alpha \in P/2} \tilde{G}_{\alpha,\alpha} = \trace{\tilde{G}}
\leq \sqrt{\trace{I}} \sqrt{\trace{\tilde{G}^2}} \leq \sqrt{r} R \,.\]
For all $\gamma \in \spt{u}$, this yields
\[
\bigl \lvert \sum_{\alpha + \beta = \gamma} \tilde{G}_{\alpha,\alpha}
\, \tilde{G}_{\beta,\beta}\bigr \rvert^{\frac{1}{2}}\leq \sum_{\alpha
  + \beta = \gamma} \frac{ \tilde{G}_{\alpha,\alpha} +
  \tilde{G}_{\beta,\beta} }{2} \leq \trace{\tilde{G}} \leq \sqrt{r} R \,,
\]
the first inequality coming again from Cauchy-Schwartz
inequality.

Thus, for all $\gamma \in \spt{u}$, one has
\[\bigl \lvert \sum_{\alpha + \beta = \gamma} \tilde{G}_{\alpha, \beta}
- (\sum_{i=1}^r s_i^2)_\gamma \bigr \rvert = \bigl \lvert \sum_{\alpha
  + \beta = \gamma} \tilde{G}_{\alpha, \beta} - \sum_{\alpha + \beta =
  \gamma} (L L^T)_{\alpha, \beta} \bigr \rvert = \bigl \lvert
\sum_{\alpha + \beta = \gamma} E_{\alpha,\beta} \bigr \rvert \,, \]
which is bounded by
\[ \frac{(r+1)2^{-\delta_c}}{1 - (r+1)2^{-\delta_c}} \sum_{\alpha +
  \beta = \gamma} |\tilde{G}_{\alpha,\alpha} \,
\tilde{G}_{\beta,\beta}|^{\frac{1}{2}} \leq
\frac{\sqrt{r}(r+1)2^{-\delta_c} \, R}{1 - (r+1)2^{-\delta_c}} \,.\]
Now, let us take the smallest $\delta$ such that
$2^{-\delta} \leq \frac{\varepsilon}{2 r} = \frac{1}{2^{N+1} r}$ as
well as the smallest $\delta_c$ such that
\[
\frac{\sqrt{r}(r+1)2^{-\delta_c} \, R}{1 - (r+1)2^{-\delta_c}} \leq
\frac{\varepsilon}{2 r}
\,, \]
that is $\delta = \lceil N + 1 + \log_2 r \rceil$ and
$\delta_c = \lceil \log_2 R + \log_2 (r+1) + \log_2 (2^{N+1} r \sqrt{r} + 1)
\rceil$.

From the previous inequalities, for all
$\gamma \in \spt{u}$, it holds that
\[|u_\gamma| = | f_\gamma - \varepsilon_{\gamma} - (\sum_{i=1}^r
s_i^2)_\gamma| \leq |f_\gamma - \varepsilon_\gamma - \sum_{\alpha +
  \beta = \gamma} \tilde{G}_{\alpha, \beta}| + |\sum_{\alpha + \beta =
  \gamma} \tilde{G}_{\alpha, \beta} - (\sum_{i=1}^r s_i^2)_\gamma |
\leq \frac{\varepsilon}{2 r} + \frac{\varepsilon}{2 r} =
\frac{\varepsilon}{r} \,.
\]
This ensures that Algorithm~$\intsos$ terminates. 

Let us note
\[\Delta(u) := \{(\alpha, \beta) : \alpha + \beta \in \spt{u} \,,
\alpha,\beta \in P/2\,, \alpha \neq \beta \} \,. \]
When terminating, the first output $\clist$ of Algorithm~$\intsos$ is a 
list of non-negative rational numbers containing the list
$[1, \dots, 1]$ of length $r$, the list
$\bigl\{\frac{|u_{\alpha + \beta}|}{2} : (\alpha, \beta) \in \Delta(u)
\bigr\}$
and the list $\{ \varepsilon_\alpha : \alpha \in \frac{P}{2} \}$.  The
second output $\slist$ of Algorithm~$\intsos$ is a list of monomials
containing the list $[s_1, \dots, s_r]$, the list
$\{ X^\alpha + \sgn{(u_{\alpha+\beta})} X^\beta : (\alpha, \beta) \in
\Delta(u) \}$
and the list $\{ X^\alpha : \alpha \in P/2 \}$.  From the output, we
obtain the following weigthed SOS decomposition
\[f = \sum_{i=1}^r s_i^2 + \sum_{\mbox{\tiny
    $(\alpha, \beta) \in \Delta(u)$ }} \dfrac{|u_{\alpha + \beta}|}{2}
(X^\alpha + \sgn{(u_{\alpha+\beta})} X^\beta)^2 + \sum_{\mbox{\tiny
    $\alpha \in \frac{P}{2}$ }} \varepsilon_\alpha X^{2 \alpha} \,.
 \]

Now, we bound the bit size of the coefficients. 
Since $r \leq \binom{n+k}{n} \leq d^n$ and
$N \leq \tau d^{\bigo{(n)}}$, one has $\delta \leq  \tau d^{\bigo{(n)}}$. Similarly, $R, \delta_c \leq \tau d^{\bigo{(n)}}$.
%
This bounds also the maximal bit size of the coefficients involved in
the approximate decomposition $\sum_{i=1}^r s_i^2$ as well the
coefficients of $u$. In the worst case, the coefficient
$\varepsilon_\alpha$ involved in the exact SOS decomposition is equal
to
$\varepsilon- \sum_{\beta \in P/2} |u_{\alpha + \beta}|/2 + u_\alpha$
for some $\alpha \in P/2$. Using again that the cardinal $r$ of $P/2$ is less than $\binom{n+k}{n} \leq d^n$, we obtain a
maximum bit size upper bounded by $\tau d^{\bigo{(n)}} $.
%
\end{proof}
\subsection{Bit complexity analysis}
\label{sec:bitop}
\begin{theorem}
\label{th:costmultivsos}
For $f$ as above, there exist $\varepsilon$, $\delta$, $R$, $\delta_c$ of bit sizes $\leq \tau d^{\bigo{(n)}}$ such that $\intsosfun{f}{\varepsilon}{\delta}{R}{\delta_c}$ runs in
boolean time $\tau^2 d^{\bigo{(n)}}$.
\end{theorem}
\begin{proof}
  We consider $\varepsilon$, $\delta$, $R$ and $\delta_c$ as in the
  proof of Proposition~\ref{th:bitmultivsos}, so that
  Algorithm~$\intsos$ only performs a single iteration within the two
  while loops before terminating. Thus, the bit size of each input
  parameter is upper bounded by $\tau d^{\bigo{(n)}}$.

Computing $\polytope(f)$ with 
the  quickhull algorithm runs in boolean time $\bigo{(V^2)}$  for
  a polytope with $V$ vertices. In our case $V \leq \binom{n+d}{n} \leq 2 d^n$,
  so that this procedure runs in boolean time
 $\bigo{(d^{2n})}$.
  Next, we investigate the computational cost of the call to $\sdp$ at
  line~\lineref{line:sdp}.  Let us note $\nsdp = r$ (resp.~$\msdp$)
  the size (resp.~number of entries) of $\tilde{G}$.  This step
  consists of solving SDP~\eqref{eq:dualsdp}, which is performed in
  $\bigo{( \nsdp^4 \log_2 (2^\tau \nsdp \, R \, 2^\delta) )}$ iterations
  of the ellipsoid method, where each iteration requires
  $\bigo{(\nsdp^2(\msdp+\nsdp) )}$ arithmetic operations over
  $\log_2 (2^\tau \nsdp \, R \, 2^\delta)$-bit numbers (see e.g.~\cite{GroetschelLovaszSchrijver93}). Since
  $\msdp, \nsdp \leq \binom{n+d}{n} \leq 2 d^n$, one has
  $\log_2 (2^\tau \nsdp \, R \, 2^\delta) \leq \tau d^{\bigo{(n)}}$,
  $\nsdp^2(\msdp+\nsdp) \leq \bigo{(\tau d^{3 n})}$ and
  $\nsdp^4 \log_2(2^\tau \nsdp \, R \, 2^\delta) \leq \tau
    d^{\bigo{(n)}}$.
  Overall, the ellipsoid algorithm runs in boolean time
  $\tau^2 d^{\bigo{(n)}}$ to compute the approximate Gram matrix
  $\tilde{G}$.
  We end with the cost of the call to $\cholesky$ at
  line~\lineref{line:chol}. Cholesky's decomposition is performed in
  $\bigo{(\nsdp^3)}$ arithmetic operations over $\delta_c$-bit
  numbers. Since $\delta_c \leq \tau d^{\bigo{(n)}}$, the function
  runs in boolean time $\tau d^{\bigo{(n)}}$.  The other elementary
  arithmetic operations performed while running Algorithm~$\intsos$
  have a negligable cost w.r.t.~to the $\sdp$ procedure. 
\end{proof}

\section{Exact Polya's representations}
\label{sec:polya}
Next, we show how to apply Algorithm~$\intsos$ to decompose positive definite forms into SOS of rational functions. 

Let $G_n := \sum_{i=1}^n X_i^2$ and
$\S^{n-1} := \{\bmx \in \R^n : G_n(\bmx) = 1 \}$ be the unit
$(n-1)$-sphere.
A positive definite form $f\in \R[X]$ is a homogeneous polynomial
which is positive over $\S^{n-1}$. For such a form, we set
\[
\varepsilon(f) := \frac{\min_{\bmx \in \S^{n-1}} f(\bmx)}{\max_{\bmx
    \in \S^{n-1}} f(\bmx)}
    \,,\]
which measures how close $f$ is to having a zero in
$\S^{n-1}$. While there is no guarantee that $f \in \Sigma[X]$, Reznick proved in~\cite{Reznick95} that for large enough $D \in \N$, $f G_n^D \in \Sigma[X]$. The proof being based on prior work by Polya~\cite{Polya}, such SOS decompositions are called {\em Polya's representations} and $D$ is called the {Polya's degree}.
%
%
Our next result states that for large enough $D \in \N$, $f G_n^D \in \mathring{\Sigma}[X]$.
\begin{lemma}
\label{th:polyaint}
Let $f$ be a positive definite form of degree $d$ in $\Z[X]$ and
$D \geq \frac{n d (d-1)}{4 \log 2 \, \varepsilon(f)} - \frac{n +
  d}{2}$.  Then $f \, G_n^{D+1} \in \mathring{\Sigma}[X]$.
\end{lemma}
\begin{proof}
  Let $P := \polytope{(f)}$ and
  $t := \sum_{\alpha \in P/2} X^{2 \alpha}$. Since $f$ is a form, then
  each term $X^{2 \alpha}$ has degree $d$, for all $\alpha \in P/2$,
  thus $t$ is a form.  First, we show that for any positive 
  $e < \frac{\min_{\bmx \in \S^{n-1}} f(\bmx)}{\max_{\bmx \in
      \S^{n-1}} t(\bmx)}$,
  the form $(f - e t)$ is positive definite: for any nonzero
  $\bmx \in \R^n$, one has
\[ f(\bmx) - e t(\bmx) = G_n(\bmx)^d [
  f\bigl(\frac{\bmx}{G_n(\bmx)}\bigr) - e
  t\bigl(\frac{\bmx}{G_n(\bmx)}\bigr)] > 0
  \]
  since $(f - e t)$ is positive on $\S^{n-1}$.  Next,
  \cite[Theorem~3.12]{Reznick95} implies that for any positive integer $D_e$ such that
  \[D_e \geq \underline{D_e} := \frac{n d (d-1)}{4 \log 2 \,
    \varepsilon(f - e t)} - \frac{n + d}{2}
 \,, \]        
  one has $(f - e t) \, G_n ^{D_e} \in \Sigma[X]$. As in the proof of
  Proposition~\ref{th:boundeps}, this yields 
  $f \, G_n ^{D_e} \in \mathring{\Sigma}[X]$. 
  
Next, with 
  $\underline{D} = \frac{n d (d-1)}{4 \log 2 \, \varepsilon(f)} -
  \frac{n + d}{2}$, 
  we prove that there exists $N \in \N$ such that for
  $e = \frac{\min_{\bmx \in \S^{n-1}} f(\bmx)}{N \max_{\bmx \in
      \S^{n-1}} t(\bmx)}$,
  $\underline{D_e} \leq \underline{D}+1$.  Since
  $f \, G_n^{D_e} \in \mathring{\Sigma}[X]$ for all
  $D_e \geq \underline{D_e}$, this will yield the desired result. For
  any $\bmx \in \S^{n-1}$, one has
\[
\min_{\bmx \in \S^{n-1}} f(\bmx) - e \max_{\bmx \in \S^{n-1}}
  t(\bmx) \leq f(\bmx) - e t(\bmx) \leq \max_{\bmx \in \S^{n-1}}
  f(\bmx) \,.
  \]
  Hence we obtain the following:
\[
\varepsilon(f - e t) \geq \frac{\min_{\bmx \in \S^{n-1}} f(\bmx) -
    e \max_{\bmx \in \S^{n-1}} t(\bmx)}{\max_{\bmx \in \S^{n-1}}
    f(\bmx)} = \varepsilon(f) \frac{N-1}{N} \,.
    \]
  Therefore, one has
  $\underline{D_e} \leq \frac{N}{N-1} \frac{n d (d-1)}{4 \log 2 \,
    \varepsilon(f)} - \frac{n + d}{2}$,
  yielding
  $\underline{D_e} - \underline{D} \leq \frac{1}{N-1} \frac{n d
    (d-1)}{4 \log 2 \, \varepsilon(f)}$.
  By choosing
  $N := \lfloor \frac{n d (d-1)}{4 \log 2 \, \varepsilon(f)} -1
  \rfloor$,
  one ensures that $\underline{D_e} - \underline{D} \leq 1$, which
  concludes the proof.
\end{proof}



%
Algorithm~$\polyasos$ takes as input $f \in \Z[X]$, finds the smallest $D \in \N$ such that $f \, G_n^D \in \mathring{\Sigma}[X]$, thanks to an oracle as in~$\intsos$. Then, $\intsos$ is applied on $f \, G_n^D$. 

{\footnotesize
\begin{algorithm}
\caption{$\polyasos$
}
\label{alg:polyasos}
\begin{algorithmic}[1]
\Require $f \in \Z[X]$, positive $\varepsilon \in \Q$, precision parameters $\delta, R \in \N$ for the SDP solver,  precision $\delta_c \in \N$ for the Cholesky's decomposition
\Ensure list $\clist$ of numbers in $\Q$ and list $\slist$ of polynomials in $\Q[X]$
\State $D := 0$
\While {$ f \, G_n^D \notin \mathring{\Sigma}[X]$} \label{line:Di} $D \gets D +1$ 
\EndWhile \label{line:Df}
\State \Return $\intsosfun{f \, G_n^D}{\varepsilon}{\delta}{R}{\delta_c}$
\end{algorithmic}
\end{algorithm}
}

\begin{example}
\label{ex:polya}
Let us apply $\polyasos$ on the perturbed Motzkin polynomial $f = (1+2^{-20}) (X_3^6 +  X_1^4 X_2^2 + X_1^2 X_2^4) - 3 X_1^2 X_2^2 X_3^2$. With $D = 1$, one has $f \, G_n = (X_1^2 + X_2^2 + X_3^2) \, f \in \mathring{\Sigma}[X]$ and $\intsos$ yields an SOS decomposition of $f \, G_n$ with $\varepsilon = 2^{-20}$, $\delta = R = 60$, $\delta_c = 10$.
\end{example}
\begin{theorem}
\label{th:polyasos}
Let $f \in \Z[X]$ be a positive definite form of degree $d$,
coefficients of bit size at most $\tau$. On input $f$,
Algorithm~$\polyasos$ terminates and outputs a weighted SOS
decomposition for $f$. The maximum bit size of its coefficients
involved and the boolean running time of the
procedure are both upper bounded by $2^{\tau d^{\bigo{(n)}}}$.
\end{theorem}
%
\begin{proof}
  By Lemma~\ref{th:polyaint}, the while loop from
  line~\lineref{line:Di} to ~\lineref{line:Df} is ensured to terminate
  for a positive integer
  $D \geq \frac{n d (d-1)}{4 \log 2 \, \varepsilon(f)} - \frac{n +
    d}{2} + 1$.
  By Proposition~\ref{th:bitmultivsos}, when applying $\intsos$ to
  $f \, G_n^D$, the procedure always terminates. The outputs are a
  list of non-negative rational numbers $[c_1,\dots,c_r]$ and a list
  of rational polynomials $[s_1,\dots,s_r]$ providing the weighted SOS
  decompositon $ f \, G_n^D = \sum_{i=1}^r c_i s_i^2$. Thus, we obtain
  $ f = \sum_{i=1}^r c_i \frac{s_i^2}{G_n^D}$, yielding the first
  claim.

  Since,
  $(X_1^2 + \dots + X_n^2)^D = \sum_{|\alpha| = D} \frac{D!}{\alpha_1!
    \cdots \alpha_n!} \, X^{2 \alpha}$,
  each coefficient of $G_n^D$ is upper bounded by
  $\sum_{|\alpha| = D} \frac{D!}{\alpha_1! \cdots \alpha_n!} = n^D$.
  Thus $\tau(f \, G_n^D) \leq \tau + D \log n$.  Using again
  Proposition~\ref{th:bitmultivsos}, the maximum bit size of the
  coefficients involved in the weighted SOS decomposition of
  $f \, G_n^D$ is upper bounded by
  $(\tau + D \log n) (d + 2D )^{\bigo{(n)}}$. Now, we derive an upper
  bound of $D$. Since $f$ is a positive form of degree $d$, one has
  \[\min_{\bmx \in \S^{n-1}} f(\bmx) = \max \{e : \forall \bmx \in
  \R^n, f(\bmx) - e G_n(\bmx)^d \geq 0 \}
 \,. \]
  Again, we rely on \cite[Theorem 14.16]{BPR06} to show that
  $\min_{\bmx \in \S^{n-1}} f(\bmx) \geq 2^{-\tau d^{\bigo{(n)}}}$.
  Similarly, we obtain
  $\max_{\bmx \in \S^{n-1}} f(\bmx) \leq 2^{\tau d^{\bigo{(n)}}}$ and
  thus $\frac{1}{\varepsilon(f)} \leq 2^{\tau d^{\bigo{(n)}}}$. We
  obtain
  $\frac{n d (d-1)}{4 \log 2 \, \varepsilon(f)} - \frac{n + d}{2} +1
  \leq 2^{ \tau d^{\bigo{(n)}}}$.
  This implies that
  $(\tau + D \log n) (d + 2 D)^{\bigo{(n)}} \leq 2^{ \tau d^{\bigo{(n)}}}$.
  From Theorem~\ref{th:costmultivsos}, the boolean running time is
  upper bounded by $(\tau + D \log n)^2 (d+2D)^{\bigo{(n)}}$, which
  ends the proof.
\end{proof}


\section{Exact Putinar's representations}
\label{sec:putinar}
We let $f, g_1,\dots, g_m$ in $\Z[X]$ of degree $\leq d$ and $\tau$ be
a bound on the bit size of their coefficients. Assume that $f$ is
positive over
$\K := \{\x \in \R^n : g_1(\x) \geq 0, \dots, g_m(\x) \geq 0 \}$ and
reaches its infimum with $f^\star := \min_{\x \in \K} f(\x) > 0$.  With
$f = \sum_{|\alpha| \leq d} f_\alpha \x^\alpha$, we set
$\|f\| := \max_{|\alpha| \leq d} \frac{f_\alpha \alpha_1!\cdots
  \alpha_n!}{|\alpha|!}$ and $g_0 := 1$. 
  
We consider the quadratic module
$\mathcal{Q}(\K) := \bigl\{ \sum_{j=0}^m \sigma_j g_j : \sigma_j \in
\Sigma[\x] \bigl\}$
and, for $D\in \N$, the $D$-truncated quadratic module
$\mathcal{Q}_D(\K) := \bigl\{ \sum_{j=0}^m \sigma_j g_j : \sigma_j \in
\Sigma[\x] \,, \ \deg(\sigma_j g_j) \leq D \bigl\}$
generated by $g_1,\dots, g_m$.  We say that $\mathcal{Q}(\K)$ is {\em
  archimedean} if $N - G_n \in \mathcal{Q}(\K)$ for some $N \in
\N$. We also assume in this section: 
\begin{assumption}
\label{hyp:arch}
The set $\K$ is a basic compact semi-algebraic set with nonempty
interior, included in $[-1, 1]^n$
and 
$\mathcal{Q}(\K)$ is archimedean.
\end{assumption}
%


Under Assumption~\ref{hyp:arch}, $f$ is positive over $\K$ only if
$f\in\mathcal{Q}_D(\K)$ for some $D \in 2 \N$
(see~\cite{Putinar1993positive}). In this case, there exists a {\em
  Putinar's representation} $f = \sum_{i=0}^m \sigma_j g_j$ with
$\sigma_j \in \Sigma[X]$ for $0\leq j\leq m$.
Let $w_j := \lceil \deg g_j / 2 \rceil$, for all
$1\leq j\leq m$.  

One can certify that $f \in \mathcal{Q}_D(\K)$ for $D = 2k$ by solving
the next SDP with
$k \geq \max \{\lceil d / 2 \rceil, w_1, \dots, w_m\}$:
\if{
\begin{equation}
\label{eq:primalsdp2}
\inf\limits_{\y} \quad  \sum_{|\alpha| \leq D} f_\alpha y_{\alpha}  \quad \text{s.t.}\quad 
 \begin{aligned}
 M_k(\y) \succeq 0 \,, \quad & \quad y_0 = 1  \,, \\
 M_{k-w_j}(g_j\, \y) \succeq 0 \,, & \quad j=1,\dots,m  \,. \\
\end{aligned}
\end{equation}
By writing $M_k(\y) = \sum_{\alpha} y_\alpha B_\alpha$ and
$M_{k-w_j}(g_j \, \y) = \sum_{\alpha } y_\alpha C_{j \alpha}$ for all
$j=1,\dots,m$ , the dual of SDP~\eqref{eq:primalsdp2} is also an SDP
given by:
}\fi
\begin{equation}
\label{eq:dualsdp2}
\begin{aligned}
  \inf\limits_{G_0, G_1,\dots,G_m \succeq 0} \quad  \trace{(G_0 \, B_0)} + \sum_{i=1}^m g_j(0) \trace{(G_j \, C_{j0})} \\
  \text{s.t.} \quad  \trace{(G_0 \, B_\gamma)} + \sum_{j=1}^m \trace{(G_j \, C_{j \gamma})} = f_\gamma \,, \quad \forall \gamma \in \N_D^n \,, \\
\end{aligned}
\end{equation}
where $B_{\gamma}$ is as for SDP~\eqref{eq:dualsdp} and $C_{j \gamma}$
has rows (resp.~ columns) indexed by $\N_{k-w_j}^{n}$ with
$(\alpha, \beta)$ entry equal to
$\sum_{\alpha+\beta+\delta = \gamma} g_{j \delta}$.
SDP~\eqref{eq:dualsdp2} is a reformulation of the problem
$\sup \{b : f - b \in \mathcal{Q}_D(\K) \}$, with optimal value
denoted by $f_D^\star$.  
Next result follows from \cite[Theorem~4.2]{Las01sos}.
\begin{theorem}
\label{th:lasoncstrained}
We use the notation and assumptions introduced above.  For
$D \in 2\N$ large enough, one has $0 < f_D^\star \leq f^\star$. In
addition, SDP~\eqref{eq:dualsdp2} has an optimal solution
$(G_0,G_1,\dots,G_m)$, yielding the following Putinar's
representation:
$f - f_D^\star = \sum_{i=1}^r \lambda_{i0} q_{i0}^2 + \sum_{i=1}^m g_j
\sum_{i=1}^{r_j} \lambda_{i j} q_{i j}^2 $
where the vectors of coefficients of the polynomials $q_{ij}$ are the
eigenvectors of $G_j$ with respective eigenvalues $\lambda_{ij}$, for
all $j=0,\dots,m$.
\end{theorem}
The complexity of Putinar's Positivstellens\"atz was analyzed in~\cite{Nie07Putinar}:

\begin{theorem}
\label{th:putinar}
With the notation and assumptions introduced above,   there
exists a real $\chi_\K > 0$ depending on $\K$ such that

(i) for all even
$D \geq \chi_\K \exp \bigl(d^2 n^d \frac{\|f\|}{f^\star} \bigr)^{\chi_\K}$,
$f \in \mathcal{Q}_D(\K)$.

(ii) for all even $D \geq \chi_\K \exp \bigl(2 d^2 n^d \bigr)^{\chi_\K}$,
$0 \leq f^\star - f_D^\star \leq \frac{6 d^3 n^{2 d}
  \|f\|}{\sqrt[\chi_\K]{\log \frac{D}{\chi_\K}}}$.
\end{theorem}
In theory, one can certify that $f$ belongs to $\mathcal{Q}_D(\K)$ for
$D = 2k$ large enough, by solving SDP~\eqref{eq:dualsdp2}. Next, we show how to ensure the existence of a strictly feasible solution for
SDP~\eqref{eq:dualsdp2} after replacing the initial set of constraints
$\K$ by $\K'$, defined as follows:
\[ 
\K':= \{ \bmx \in \K : 1 - \bmx^{2 \alpha} \geq 0 \,, \forall \alpha \in
\N^n_k \}
\,. \]
We first give a lower bound for $f^\star$.
%
\begin{proposition}
\label{th:lowerboundcube}
With the above notation and assumptions, one has:
\[
f^\star \geq 2^{-(\tau+ d + d \log_2 n
  + 1)d^{n+1}} d^{-(n+1)d^{n+1}} = 2^{-\tau d^{\bigo{(n)}}}
\,.  \]
\end{proposition}
\begin{proof}
  Let $Y=(Y_1, \ldots, Y_n)$ and $\tilde{f}\in \Z[Y]$ be the
  polynomial obtained by replacing $Y_i$ by $2n Y_i-1$ in $f$.  Note
  that if $\bmx=(x_1,\dots, x_n) \in \K \subseteq [-1, 1]^n$, then
  $\bmy = \left (\left (\frac{x_i+1}{2n}\right )\right )_{1\leq i
    \leq n}$
  lies in the standard simplex $\Delta_n$, so the polynomial
  $\tilde{f}$ takes only positive values over 
  $\Delta_n$.
  Since $x_i = 2 n y_i - 1$ and $(2 n - 1)^d \leq (2 n)^d$, the
  polynomial $\tilde{f}$ has coefficients of bit size at most
  $\tau + d + d \log_2 n$. Then, the desired result follows
  from~\cite[Theorem~1]{Jeronimo10}, stating that
  $\min_{\bmy \in \Delta_n} \tilde{f}(\bmy) >
  2^{-(\tau{(\tilde{f})}+1)d^{n+1}} d^{-(n+1)d^{n+1}}$.
\end{proof}
\begin{theorem}
\label{th:putinarint}
We use the notation and assumptions introduced above.  There exists
$D \in 2 \N$ such
that:\\
(i) $f \in \mathcal{Q}_D(\K)$ with the representation
\[f = f_D^\star + \sum_{j=0}^m \sigma_j g_j\] 
for $f_D^\star > 0$,
$\sigma_j \in \Sigma[X]$ with $\deg (\sigma_j g_j) \leq D$ for all
$j=0,\dots,m$.\\
(ii) $f \in \mathcal{Q}_D(\K')$ with the representation
\[f = \sum_{j=0}^m \mathring{\sigma_j} g_j + \sum_{|\alpha| \leq k}
c_\alpha (1 - X^{2 \alpha}) \]
for $\mathring{\sigma_j} \in \mathring{\Sigma}[X]$ with
$\deg (\mathring{\sigma_j} g_j) \leq D$, for all $j=0,\dots,m$, and
some sequence of positive numbers $(c_\alpha)_{|\alpha| \leq k}$.\\
(iii) There exists a real $C_\K > 0$ depending on $\K$ and
$\varepsilon = \frac{1}{2^N}$ with positive $N \in \N$ such
that
$f - \varepsilon \sum_{|\alpha| \leq k} X^{2 \alpha} \in
\mathcal{Q}_D(\K')$
and $N \leq 2^{\tau d^{n C_\K}}$, where $\tau$ is the maximal bit size
of the coefficients of $f,g_1,\dots,g_m$.
\end{theorem}
\begin{proof}
  Let $\chi_\K$ be as in Theorem~\ref{th:putinar} and $D = 2 k$ be the
  smallest integer larger than $\underline{D}$ given by:
  \[\underline{D} := \max\{\chi_\K \exp \bigl(\frac{12 d^3 n^{2d}
    \|f\|}{f^\star}\bigr)^{\chi_\K}, \chi_\K \exp(2 d^2 n^d)^{\chi_\K} \}  \,.\]
  Theorem~\ref{th:putinar} implies $f \in \mathcal{Q}_D(\K)$
  and
  $f^\star - f_{D}^\star \leq \frac{6 d^3 n^{2 d}
    \|f\|}{\sqrt[\chi_\K]{\log \frac{{D}}{\chi_\K}}} \leq
  \frac{f^\star}{2}$.

  (i) This yields the representation
  $f - f_D^\star = \sum_{j=0}^m \sigma_j g_j$, with
  $f_D^\star \geq \frac{f^\star}{2} > 0$, $\sigma_j \in \Sigma[X]$ and
  $\deg (\sigma_j g_j) \leq D$ for all $j=0,\dots,m$.

  (ii) For $1\leq j\leq m$, let us define
  $t_j := \sum_{|\alpha| \leq k - w_j} X^{2 \alpha}$,
  $t_0 := \sum_{|\alpha| \leq k} X^{2 \alpha}$ and
  $t := \sum_{j=0}^m t_j g_j$.  For a given $\nu > 0$, we use the
  perturbation polynomial
  $-\nu t = -\nu \sum_{|\gamma|\leq D}t_\gamma X^\gamma$. For each
  term $-t_\gamma X^\gamma$, one has $\gamma = \alpha + \beta$ with
  $\alpha, \beta \in \N^n_k$, thus
  $-t_\gamma X^\gamma = |t_\gamma|(-1 + \frac{1}{2} (1 - X^{2
    \alpha}) + \frac{1}{2} (1 - X^{2 \beta}) + \frac{1}{2} (X^\alpha -
  \sgn{(t_\gamma)} X^\beta)^2)$.
As in the proof of Proposition~\ref{th:bitmultivsos}, let us note
  $\Delta(t) := \{(\alpha, \beta) : \alpha + \beta \in \spt{t} \,,
  \alpha,\beta \in \N^n_k \,, \alpha \neq \beta \}$.
  Hence, there exist $d_\alpha \geq 0$ for all $\alpha \in \N^n_k$ such
  that 
  \[f = f - \nu t + \nu t =
   f_D^\star - \sum_{|\gamma|\leq D} \nu |t_\gamma| + \sum_{j=0}^m
  \sigma_j g_j + \nu t + \sum_{{ |\alpha| \leq k }} d_\alpha (1 -
  X^{2 \alpha}) + \nu \sum_{ (\alpha, \beta) \in \Delta(t) }
  \dfrac{|t_{\alpha + \beta}|}{2} (X^\alpha -
  \sgn{(t_{\alpha+\beta})} X^\beta)^2  \,.\]
Since one has not necessarily $d_\alpha > 0$ for all $\alpha \in \N^n_k$, we now explain how to 
  handle the case when $d_\alpha = 0$ for 
  $\alpha \in \N^n_k$. We write
\begin{align*}
- \sum_{|\gamma|\leq D} \nu |t_\gamma| + \sum_{{ |\alpha| \leq k }}
  d_\alpha (1 - X^{2 \alpha}) = & - \sum_{|\gamma|\leq D} \nu |t_\gamma|
  - \sum_{\alpha : d_\alpha = 0} \nu 
  + \sum_{\alpha : d_\alpha = 0}
  \nu (1 - X^{2\alpha}) 
  + \sum_{\alpha : d_\alpha = 0} \nu X^{2
    \alpha} \\
    & + \sum_{{ |\alpha| : d_\alpha = 0 }} d_\alpha (1 - X^{2
    \alpha}) + \sum_{{ |\alpha| : d_\alpha > 0 }} d_\alpha (1 - X^{2
    \alpha}) \,. 
    \end{align*}
For $\alpha \in \N^n_k$, we define $c_\alpha := \nu$ if
  $d_\alpha = 0$ and $c_\alpha := d_\alpha$ otherwise, 
  $a := \sum_{|\gamma|\leq D} |t_\gamma| + \sum_{\alpha : d_\alpha =
    0} 1$,
  $\mathring{\sigma}_j := \sigma_j + \nu t_j$, for each $j=1,\dots,m$
  and
  \[\mathring{\sigma}_0 := f_D^\star - \nu a + \sigma_0 + \nu t_0 + \nu
  \sum_{ (\alpha, \beta) \in \Delta(t) } \dfrac{|t_{\alpha +
      \beta}|}{2} (X^\alpha - \sgn{(t_{\alpha+\beta})} X^\beta)^2 +
  \sum_{\alpha : d_\alpha = 0} \nu X^{2 \alpha} \,. \]
    So, there
  exists a sequence of positive numbers $(c_\alpha)_{|\alpha|\leq k}$
  such that
  \[f = \sum_{j=0}^m \mathring{\sigma_j} g_j + \sum_{{ |\alpha| \leq k
    }} c_\alpha (1 - X^{2 \alpha}) \,. \]
%
%
  Now, let us select $\nu := \frac{1}{2^M}$ with $M$ being the
  smallest positive integer such that
  $0 < \nu \leq \frac{f_D^\star}{2 a}$. 
 This implies the existence of a positive definite Gram matrix for $\mathring{\sigma_0}$, thus by Theorem~\ref{th:intsospdGram}, 
  $\mathring{\sigma_0} \in \mathring{\Sigma}[X]$. Similarly, for
  $1\leq j\leq m$, $\mathring{\sigma}_j$  belongs to
  $\mathring{\Sigma}[X]$, which proves the second claim.

  (iii) Let $N := M+1$ and
  $\varepsilon := \frac{1}{2^N} = \frac{\nu}{2}$. 
  One has
  \[f - \varepsilon \sum_{|\alpha| \leq k} X^{2 \alpha} = f -
  \varepsilon t_0 = \mathring{\sigma}_0 - \varepsilon t_0 +
  \sum_{j=1}^m \mathring{\sigma}_j g_j + \sum_{|\alpha| \leq k}
  c_\alpha (1 - X^{2 \alpha}) \,. \]
  Thus, $\sigma_0 + (\nu - \varepsilon) t_0 \in
  \mathring{\Sigma}[X]$.
  This implies that
  $\mathring{\sigma}_0 - \varepsilon t_0 \in \mathring{\Sigma}[X]$ and
  $f - \varepsilon t_0 \in \mathcal{Q}_D(\K')$. Next, we derive a
  lower bound of $\frac{f^\star_D}{a}$.
  Since
  $t =\sum_{|\alpha|\leq k} X^{2\alpha} + \sum_{j=1}^m g_j
  \sum_{|\alpha| \leq k - w_j} X^{2\alpha}$,
  one has
  $\sum_{|\gamma| \leq D} |t_\gamma| \leq 2^\tau (m+1)
  \binom{n+D}{n}$.
  This implies that
  \[a \leq 2^{\tau} (m+1) \binom{n+D}{n} + \binom{n+k}{k} \leq 2^{\tau}
  (m+2) \binom{n+D}{n}
 \,. \]
  Recall that $\frac{f^\star}{2} \leq f^\star_D$, implying
  \[ \frac{f^\star_{{D}}}{a} \geq \frac{f^\star }{2^{\tau+1} (m+2)
    \binom{n+{D}}{n}} \geq \frac{1}{(m+2) 2^{\tau d^{\bigo{(n)}}}
    {D}^n} \,, \]
  where the last inequality follows from
  Theorem~\ref{th:lowerboundcube}.  Let us now give an upper bound of
  $\log_2 D$. First, note that for all $\alpha \in \N^n$,
  $\frac{|\alpha|!}{\alpha_1!\cdots\alpha_n!} \geq 1$, thus
  $\|f\| \leq 2^\tau$.
  Since $D$ is the smallest even integer larger than $\underline{D}$,
  one has
  \[ 
  \log_2 D \leq 1 + \log_2 \underline{D} \leq 1 + \log \chi_\K + (12 d^3
  n^{2 d} 2^\tau 2^{ \tau d^{\bigo{(n)}}})^{\chi_\K}
 \,. \]
  Next, since $N$ is the smallest integer such that
  $\varepsilon = \frac{1}{2^N} = \frac{\nu}{2} \leq \frac{f^\star_D}{2
    a}$,
  it is enough to take
  \[
  N \leq 1 + \log_2 (m+2) + \tau d^{\bigo{(n)}} + n \log_2 D \leq
  2^{\tau d^{n C_\K}}
  \]
  for some real $C_\K > 0$ depending on $\K$, the desired result.
\end{proof}

\begin{algorithm}
  \caption{$\putinarsos$.
    }
\label{alg:putinarsos}
\begin{algorithmic}[1]
\Require $f \in \Z[X]$, $\K := \{\x \in \R^n : g_1(\x) \geq 0, \dots, g_m(\x) \geq 0 \}$ with $g_1,\dots,g_m \in \Z[X]$, positive $\varepsilon \in \Q$, precision parameters $\delta, R \in \N$ for the SDP solver 
,  precision $\delta_c \in \N$ for the Cholesky's decomposition
\Ensure lists $\clist_0,\dots,\clist_m, \calpha$ of numbers in $\Q$ and lists $\slist_0,\dots,\slist_m$ of polynomials in $\Q[X]$
\State $k \gets \max \{ \lceil d / 2 \rceil,w_1, \dots,  w_m\}$, $D \gets 2 k$, $g_0 := 1$
\While {$f  \notin \mathcal{Q}_D(\K)$} \label{line:DPuti} 
 $k \gets k + 1$, $D \gets D + 2$
\EndWhile \label{line:DPutf}
\State $P := \N^n_D$,   $\K':= \{ \x \in \K : 1 - \x^{2 \alpha} \geq 0 \,, \forall \alpha \in \N^n_k \}$ \label{line:npP}
\State $t := \sum_{\alpha \in P/2} X^{2 \alpha}$, $f_\varepsilon \gets f - \varepsilon t$
\While {$f_\varepsilon \notin  \mathcal{Q}_D(\K')$} \label{line:epsiP}
 $\varepsilon \gets \frac{\varepsilon}{2}$, $f_\varepsilon \gets f - \varepsilon t$
\EndWhile \label{line:epsfP}
\State ok := false
\While {not ok} \label{line:deltaiP}
\State $[\tilde{G_0},\dots,\tilde{G}_{m},  \tilde{\lambda}_0,\dots,\tilde{\lambda}_{m},(\tilde{c}_\alpha)_{|\alpha| \leq k}], \gets \sdpconfun{f_\varepsilon}{ \delta}{R}{\K'}$ \label{line:sdpcon}
\State $\calpha \gets (\tilde{c}_\alpha)_{|\alpha| \leq k}$
\For {$j \in \{0,\dots,m\}$}
\State $(s_{1j},\dots,s_{r_j j}) \gets \choleskyfun{\tilde{G}_j}{\tilde{\lambda}_j}{\delta_c}$, $\tilde{\sigma}_j := \sum_{i=1}^{r_j} s_{i j}^2 $ \label{line:cholcon}
\State $\clist_j \gets [1,\dots,1] $, $\slist_j \gets [s_{1j},\dots,s_{r_j j}]$
\EndFor
\State $u \gets f_\varepsilon - \sum_{j=0}^m   \tilde{\sigma_j} \, g_j - \sum_{|\alpha| \leq k} \tilde{c}_\alpha (1 - X^{2 \alpha})$
\For {$\alpha \in P/2$}  $\varepsilon_{\alpha} := \varepsilon$
\EndFor
\State $\clist, \slist,(\varepsilon_\alpha) \gets \absorbfun{u}{P}{(\varepsilon_\alpha)}{\clist}{\slist}$ \label{line:absorbP}
\If {$\min_{\alpha \in P/2} \{ \varepsilon_\alpha \} \geq 0$} ok := true 
\Else $\ \delta \gets 2 \delta$, $R \gets 2 R$, $\delta_c \gets 2 \delta_c$
\EndIf
\EndWhile \label{line:deltafP}
\For {$\alpha \in P/2$}  
\State $\clist_0 \gets  \clist_0 \cup \{ \varepsilon_\alpha \}$, $\slist_0 \gets  \slist_0 \cup \{ \x^\alpha \}$
\EndFor
\State \Return $\clist_0,\dots,\clist_m, \calpha,  \slist_0,\dots,\slist_m$
\end{algorithmic}
\if{
\begin{algorithm}
\caption{$\putinarsos$: algorithm to compute exact rational Putinar's representations of polynomials positive over basic compact semialgebraic sets.}
\label{alg:putinarsos}
\begin{algorithmic}[1]
\Require $f, g_1,\dots,g_m \in \Z[\x]$, $\varepsilon \in \Q$ such that $0 < \varepsilon$, precision parameters $\delta, R \in \N$ for the SDP solver,  precision $\delta_c \in \N$ for the Cholesky's decomposition
\Ensure list $\clist$ of numbers in $\Q$ and list $\slist$ of polynomials in $\Q[\x]$
\State $k \gets \max \{ \lceil d / 2 \rceil,w_1, \dots,  w_m\}$, $D \gets 2 k$
\State $\K := \{\x \in \R^n : g_1(\x) \geq 0, \dots, g_m(\x) \geq 0 \}$, $g_0 := 1$
\While {$f  \notin \mathcal{Q}_D(\K)$} \label{line:DPuti} 
\State $k \gets k + 1$, $D \gets D + 2$
\EndWhile \label{line:DPutf}
\State $\K':= \{ \x \in \K : 1 - \x^{2 \alpha} \geq 0 \,, \forall \alpha \in \N^n_k \}$, ok := false
\While {not ok} \label{line:KPuti} 
\State $[\tilde{G_0},\dots,\tilde{G}_{m},  \tilde{\lambda}_0,\dots,\tilde{\lambda}_{m},(c_\alpha)_{|\alpha| \leq k}], \gets \sdpconfun{f_\varepsilon}{ \delta}{R}{\K'}$ \label{line:sdpcon}
\For {$j \in \{0,\dots,m\}$}
\State $(s_{1j},\dots,s_{r_j j}) \gets \choleskyfun{\tilde{G}_j}{\tilde{\lambda}_j}{\delta_c}$, $\mathring{\sigma}_j := \sum_{i=1}^{r_j} s_{i j}^2 $ \label{line:cholcon}
\EndFor
\If {$\min_{0\leq j \leq m} \lambda_j > 0$ \textbf{and} $\min_{|\alpha|\leq k} c_\alpha > 0$ } 
ok := true 
\Else $\ \delta \gets 2 \delta$, $R \gets 2 R$, $\delta_c \gets 2 \delta_c$
\EndIf
\EndWhile \label{line:KPutf}
\State $\mathring{\sigma}_0 := f - \sum_{j=1}^m \mathring{\sigma}_j  g_j - \sum_{|\alpha| \leq k} c_\alpha (1 - X^{2 \alpha})$
\State \Return $\intsosfun{\mathring{\sigma}_0}{\varepsilon}{\delta}{R}{\delta_c}$
\end{algorithmic}
\end{algorithm}
}\fi
\end{algorithm}

We can now present Algorithm~$\putinarsos$. For $f \in \Z[X]$ positive
over a basic compact semi-algebraic set $S$ satisfying
Assumption~\ref{hyp:arch}, the first loop 
outputs the smallest positive integer
$D = 2k$ such that $f \in \mathcal{Q}_D(\K)$. Then the procedure is
similar to $\intsos$. As for the first loop of
$\intsos$, the loop from line~\lineref{line:epsiP} to
line~\lineref{line:epsfP} allows to obtain a perturbed polynomial
$f_\varepsilon \in \mathcal{Q}_D(\K')$, with
$\K' := \{ \bmx \in \K : 1 - \bmx^{2 \alpha} \geq 0 \,, \forall \alpha
\in \N^n_k \}$.
Then one solves SDP~\eqref{eq:dualsdp2} with the $\sdpcon$ procedure
and performs Cholesky's decomposition to obtain an approximate
Putinar's representation of $f_\varepsilon = f - \varepsilon t$
and a remainder $u$. Next, we apply the $\absorb$ subroutine as 
in $\intsos$. The rationale is that with large
enough precision parameters for the procedures $\sdpcon$ and
$\cholesky$, one finds an exact weighted SOS decomposition of
$u + \varepsilon t$, which yields in turn an exact Putinar's
representation of $f$ in $\mathcal{Q}_D(\K')$ with rational
coefficients.
\begin{example}
\label{ex:putinar}
Let us apply $\putinarsos$ to $f =-X_1^2 - 2 X_1 X_2 - 2 X_2^2 + 6$,
$S := \{(x_1,x_2) \in \R^2 : 1 - x_1^2 \geq 0, 1 - x_2^2 \geq 0\}$
and the same precision parameters as in Example~\ref{ex:intsos}. The
first and second loop yield $D = 2$ and 
$\varepsilon = 1$. After running $\absorb$, we obtain the exact
Putinar's representation 
$f = \frac{23853407}{292204836} + \frac{23}{49} X_1^2 +
\frac{130657269}{291009481} X_2^2 + \frac{1}{2442^2} + (X_1-X_2)^2 +
(\frac{X_2}{2437})^2+(\frac{11}{7})^2 (1-X_1^2) + (\frac{13}{7})^2
(1-X_2^2)$.
\end{example}

\begin{theorem}
\label{th:putinarsos}
We use the notation and assumptions introduced above.  
For some $C_\K > 0$ depending on $\K$, there exist $\varepsilon$,
$\delta$, $R$, $\delta_c$ and $D = 2 k$ of bit sizes less than 
$\bigo{(2^{\tau d^{n C_\K}})}$ for which
$\putinarsosfun{f}{\K}{\varepsilon}{\delta}{R}{\delta_c}$ terminates
and outputs an exact Putinar's representation with rational
coefficients of $f \in \mathcal{Q}(\K')$, with
$\K' := \{ \bmx \in \K : 1 - \bmx^{2 \alpha} \geq 0 \,, \forall \alpha
\in \N^n_k \}$.
The maximum bit size of these coefficients is bounded by
$\bigo{(2^{\tau d^{n C_\K}})}$ and the procedure runs in boolean time
$\bigo{\bigl(2^{2^{\tau d^{n C_\K}}} \bigr)}$.
\end{theorem}
\begin{proof}
  The loops going from line~\lineref{line:DPuti} to
  line~\lineref{line:DPutf} and from line~\lineref{line:epsiP} to
  line~\lineref{line:epsfP} always terminate as respective
  consequences of Theorem~\ref{th:putinarint}~(i) and
  Theorem~\ref{th:putinarint}~(iii) with $D \leq 2^{\tau d^{n C_\K}}$,
  $\varepsilon = \frac{1}{2^N}$, $N \leq 2^{\tau d^{n C_\K}}$, for
  some real $C_\K > 0$ depending on $\K$.

  What remains to prove is similar to
  Proposition~\ref{th:bitmultivsos} and
  Theorem~\ref{th:costmultivsos}.  Let $\nu$,
  $\mathring{\sigma}_0,\dots,\mathring{\sigma}_m, (c_\alpha)_{|\alpha|
    \leq k}$
  be as in the proof of Theorem~\ref{th:putinarint}. Note that $\nu$
  (resp.~$\varepsilon-\nu$) is a lower bound of the smallest
  eigenvalues of any Gram matrix associated to $\mathring{\sigma}_j$
  (resp.~$\mathring{\sigma}_0$) for $1\leq j\leq m$. In addition,
  $c_\alpha \geq \nu$ for all $\alpha \in \N^n_k$.  When the $\sdp$
  procedure at line~\lineref{line:sdpcon} succeeds, the matrix
  $\tilde{G}_j$ is an approximate Gram matrix of the polynomial
  $\mathring{\sigma}_j$ with $\tilde{G}_j \succeq 2^{\delta} I $,
  $\sqrt{\trace{(\tilde{G}_j^2)}} \leq R$, we obtain a positive
  rational approximation $\tilde{\lambda}_j \geq 2^{-\delta}$ of the
  smallest eigenvalue of $\tilde{G}_j$, $\tilde{c_\alpha}$ is a
  rational approximation of $c_\alpha$ with
  $\tilde{c_\alpha} \geq 2^{-\delta}$, and $\tilde{c_\alpha} \leq R$,
  for all $j=0,\dots,m$ and $\alpha \in \N^n_k$.  This happens when
  $2^{-\delta} \leq \varepsilon$ and
  $2^{-\delta} \leq \varepsilon - \nu$, thus for
  $\delta = \bigo{(2^{\tau d^{n C_\K}})}$. As in the proof of
  Proposition~\ref{th:boundR}, we derive a similar upper bound of $R$
  by a symmetric argument while considering a Putinar representation
  of $\overline{f}_D - f \in \mathcal{Q}_D(\K')$, where
  $\overline{f}_D := \inf \{b : b - f \in \mathcal{Q}_D(\K) \}$.  As
  for the second loop of Algorithm~$\intsos$, the third loop of
  $\putinarsos$ terminates when the remainder polynomial
  $u = f_\varepsilon - \sum_{j=0}^m \tilde{\sigma_j} \, g_j -
  \sum_{|\alpha| \leq k} \tilde{c}_\alpha (1 - X^{2 \alpha})$
  satisfies $|u_\gamma| \leq \frac{\varepsilon}{r_0}$, where
  $r_0 = \binom{n+k}{n}$ is the size of $P/2 = \N^n_k$. As in the
  proof of Proposition~\ref{th:bitmultivsos}, one can show that this
  happens when $\delta$ and $\delta_c$ are large enough.

  To bound the precision $\delta_c$ required for Cholesky's
  decomposition, we do as in the proof of
  Proposition~\ref{th:bitmultivsos}.
%
  The difference now is that there are $m + \binom{n+k}{k} = m + r_0$
  additional terms in each equality constraint of
  SDP~\eqref{eq:dualsdp2}, by comparison with SDP~\eqref{eq:dualsdp}.
  Thus, we need to bound for all $j=1\,\dots,m$, $\alpha \in \N^n_k$
  and $\gamma \in \spt{u}$ each term
  $|\trace{(\tilde{G_j} C_{j \gamma})} - (g_j \tilde{\sigma})_\gamma|$
  related to the constraint $g_j \geq 0$ as well as each term (omitted
  for conciseness) involving $\tilde{c}_\alpha$ related to the
  constraint $1 - X^{2 \alpha} \geq 0$.
  By using the fact that
  \[
  \trace{(\tilde{G}_j C_{j\gamma})} = \sum_\delta g_{j \delta}
  \sum_{\alpha + \beta + \delta = \gamma} \tilde{G}_{j \alpha,\beta}
 \,, \]
  we obtain
  \[
  |\trace{(\tilde{G_j} C_{j \gamma})} - (g_j \tilde{\sigma})_\gamma|
  \leq \sum_\delta |g_{j \delta}| \frac{\sqrt{r_j}(r_j+1)2^{-\delta_c}
    \, R}{1 - (r_j+1)2^{-\delta_c}} 
    \,, \]    
  where $r_j$ is the size of $\tilde{G_j}$.
  Note that the size $r_0$ of the matrix $\tilde{G_0}$ satisfies
  $r_0 \geq r_j$ for all $j=1,\dots,m$. In addition, $\deg g_j \leq D$
  implies
  \[
  \sum_\delta |g_{j \delta}| \leq \binom{n+\deg g_j}{n} 2^\tau \leq
  \binom{n+D}{n} 2^\tau \leq D^n 2^{\tau+1}
 \,. \]
  This yields an upper bound of 
  $D^n 2^{\tau+1} \frac{\sqrt{r_0}(r_0+1)2^{-\delta_c} \, R}{1 -   (r_0+1)2^{-\delta_c}}$.
  We obtain a similar bound (omitted for conciseness) for each term
  involving $\tilde{c}_\alpha$.

  Then, we take the smallest $\delta$ such that
  $2^{-\delta} \leq \frac{\epsilon}{2 r_0}$ and the smallest
  $\delta_c$ such that
  \[ D^n 2^{\tau} \frac{\sqrt{r_0}(r_0+1)2^{-\delta_c} \, R}{1 -
    (r_0+1)2^{-\delta_c}} \leq \frac{\varepsilon}{2 r_0 ((m+1) +
    r_0)} 
   \,, \]
  which ensures that $\putinarsos$ terminates. One proves that the
  procedure outputs an exact Putinar's representation of
  $f \in \mathcal{Q}(\K')$ with rational coefficients of maximum bit
  size bounded by $\bigo{(2^{\tau d^{n C_\K}})}$.

  As in the proof of Theorem~\ref{th:costmultivsos}, let $\nsdp$ be
  the sum of the sizes of the matrices involved in
  SDP~\eqref{eq:dualsdp2} and $\msdp$ be the number of entries. Note
  that 
  \[\nsdp \leq (m+1) r_0 + r_0 \leq (m+2) \binom{n+D}{n}
  \]
  and
  $\msdp := \binom{n+D}{n}$.  To bound the boolean running time, we
  consider the cost of solving SDP~\eqref{eq:dualsdp2}, which is
  performed in
  $\bigo{( \nsdp^4 \log_2(2^\tau \nsdp \, R \, 2^\delta) )}$
  iterations of the ellipsoid method, where each iteration requires
  $\bigo{(\nsdp^2(\msdp+\nsdp) )}$ arithmetic operations over
  $\log_2(2^\tau \nsdp \, R \, 2^\delta)$-bit numbers.  Since $\msdp$
  is bounded by $\binom{n+D}{n} \leq 2 D^n$ and
  $\log_2 D = \bigo{({2^{\tau d^{n C_\K}}} )}$, one has
  $\msdp = \bigo{\bigl(2^{2^{\tau d^{n C_\K}}} \bigr)}$. We obtain the
  same bound for $\nsdp$, which ends the proof.
\end{proof}
The complexity is polynomial in the degree $D$ of the representation,
often close in practice to the degrees of the involved polynomials, as
shown in Section~\ref{sec:benchs}.


\section{Practical experiments}
\label{sec:benchs}
We provide practical performance results for Algorithms~$\intsos$,
$\polyasos$ and $\putinarsos$.  These are implemented in a library,
called $\multivsos$, written in Maple. More details about installation
and benchmark execution are given on the two webpages dedicated to
univariate\footnote{\url{https://github.com/magronv/univsos}} and
multivariate\footnote{\url{https://github.com/magronv/multivsos}}
polynomials.  This tool is available within the RAGlib Maple
package\footnote{\url{http://www-polsys.lip6.fr/~safey/RAGLib/}}.  
All
results were obtained on an Intel Core i7-5600U CPU (2.60 GHz) with
16Gb of RAM. We use the Maple~\texttt{Convex}
package\footnote{\url{http://www-home.math.uwo.ca/~mfranz/convex}} to
compute Newton polytopes. 
Our subroutine $\sdp$ relies on the
arbitrary-precision solver SDPA-GMP~\cite{Nakata10GMP} and the
$\cholesky$ procedure is implemented with the
function~\texttt{LUDecomposition} available within Maple. Most of the
time is spent in the $\sdp$ procedure for all benchmarks.

In Table~\ref{table:bench1}, we compare the performance of
$\multivsos$ for nine univariate polynomials being positive over
compact intervals. More details about these benchmarks are given
in~\cite[Section~6]{Chevillard11} and~\cite[Section~5]{univsos}. In
this case, we use $\putinarsos$. The main difference is that we use
SDP in $\multivsos$ instead of complex root isolation in
$\univsostwo$. The results emphasize that $\univsostwo$ performs
better and provides more concise SOS certificates, especially for high
degrees (see e.g.~\# 5). For \# 3, we were not able to obtain a
decomposition within a day of computation, as meant by the symbol $-$
in the corresponding column entries. Large values of $d$ and $\tau$
require more precision. The values of $\varepsilon$, $\delta$ and
$\delta_c$ are respectively between $2^{-80}$ and $2^{-240}$, 30 and
100, 200 and 2000.

Next, we compare the performance of $\multivsos$ with other tools in
Table~\ref{table:bench2}. The two first benchmarks are built from the
polynomial
$f = (X_1^2+1)^2+(X_2^2+1)^2 + 2 (X_1+X_2+1)^2 - 268849736/10^8$
from~\cite[Example~1]{Las01sos}, with $f_{12} := f^3$ and
$f_{20} := f^5$. For these two benchmarks, we apply $\intsos$. We use
$\polyasos$ to handle $M_{20}$ (resp.~$M_{100}$), obtained as in
Example~\ref{ex:polya} by adding $2^{-20}$ (resp.~$2^{-100}$) to the
positive coefficients of the Motzkin polynomial and $r_{i}$, which is
a randomly generated positive definite quartic with $i$ variables. We
implemented in Maple the projection and rounding algorithm
from~\cite{PaPe08} also relying on SDP, denoted by~$\PP$. For
$\multivsos$, the values of $\varepsilon$, $\delta$ and $\delta_c$ lie
between $2^{-100}$ and $2^{-10}$, 60 and 200, 10 and 60.
We compare with~$\raglib$ based on critical points and
the~\texttt{SamplePoints} procedure (abbreviated as~$\cad$) based on
CAD, both available in Maple. While these methods outperform the two
SDP-based algorithms for examples with $n \leq 3$, they are less
efficient for larger examples such as $r_6^2$ and suffer from a severe
computational burden when $n \geq 8$. An additional drawback is that
they do not provide non-negativity certificates. However, note that
they can solve less restrictive problems, involving positive
semidefinite forms or non-negative polynomials. 

{As shown in~\cite{KLYZ08}, SDP-based methods may provide exact
  certificates even in such cases and can be extended to rational
  functions. The algorithms we developed in this paper are unable to
  handle such cases. }
%
In most cases, $\multivsos$ is more efficient than $\PP$ and outputs
more concise representations. The reason is that $\multivsos$ performs
approximate Cholesky's decompositions while $\PP$ computes exact
$L D L^T$ decompositions of Gram matrices obtained after the two steps
of rounding and projection.  
Note that we could not solve the examples of
Table~\ref{table:bench2} with less precision.
%
%
%

{\begin{table}[!t]
\begin{center}
\caption{\footnotesize $\multivsos$ vs $\univsostwo$ \cite{univsos} for benchmarks from~\cite{Chevillard11}.}
\begin{tabular}{lcr|rr|rr}
\hline
\multirow{2}{*}{Id} & \multirow{2}{*}{$d$} & \multirow{2}{*}{$\tau$ (bits)} & \multicolumn{2}{c|}{$\multivsos$} & \multicolumn{2}{c}{$\univsostwo$} \\
 & & & $\tau_1$ (bits) & $t_1$ (s)  & $\tau_2$ (bits) & $t_2$ (s)  \\
\hline  
\# 1 & 13 & 22 682 & 387 178 & 0.84 & 51 992 & 0.83  \\
\# 3 & 32 & 269 958 & $-$ & $-$ & 580 335 & 2.64  \\
\# 4 & 22 & 47 019 &  1 229 036 & 2.08  & 106 797 & 1.78 \\
\# 5 & 34 & 117 307 & 10 271 899 & 69.3  &   265 330 & 5.21 \\
\# 6 & 17 & 26 438 & 713 865 & 1.15  & 59 926 & 1.03 \\
\# 7 & 43 & 67 399 & 10 360 440 &  16.3  & 152 277 & 11.2  \\
\# 8 & 22 & 27 581 & 1 123 152 & 1.95 & 63 630 & 1.86 \\
\# 9 & 20 & 30 414 & 896 342 & 1.54  & 68 664 & 1.61 \\
\# 10 & 25 & 42 749 & 2 436 703 & 3.02  &  98 926 & 2.76  \\
\hline
\end{tabular}
\label{table:bench1}
\end{center}
\end{table}} 
{\begin{table}[!ht]
\begin{center}
\caption{$\multivsos$ vs $\PP$ \cite{PaPe08} vs $\raglib$ vs $\cad$ for $n$-variate polynomials of degree $d$ (Polya).
}
\begin{tabular}{lrr|rr|rr|c|c}
\hline
\multirow{2}{*}{Id} & \multirow{2}{*}{$n$} & \multirow{2}{*}{$d$}  & \multicolumn{2}{c|}{$\multivsos$} & \multicolumn{2}{c|}{$\PP$} & $\raglib$ & $\cad$ \\
 & & & $\tau_1$ (bits) & $t_1$ (s)  & $\tau_2$ (bits) & $t_2$ (s) & $t_3$ (s) & $t_4$ (s) \\
\hline  
$f_{12}$ & 2 & 12 & 162 861 & 5.96 & 5 185 020 & 6.92 & 0.15 & 0.07 \\
$f_{20}$ & 2 & 20 & 745 419 & 110. & 78 949 497 & 141. & 0.16 & 0.03 \\
$M_{20}$ & 3 & 8 & 4 695 & 0.18 & 3 996 & 0.15 & 0.13 & 0.05 \\
$M_{100}$ & 3 & 8 & 17 232 & 0.35 & 18 831 & 0.29 & 0.15 & 0.03 \\
$r_2$ & 2 & 4 & 1 866 & 0.03 & 1 031 & 0.04 & 0.09 & 0.01 \\
$r_4$ & 4 & 4 & 14 571 & 0.15 &  47 133 & 0.25 & 0.32 & $-$ \\
$r_6$ & 6 & 4 & 56 890 & 0.34 & 475 359 & 0.54 & 623. & $-$ \\
$r_8$ & 8 & 4 & 157 583 & 0.96 & 2 251 511 & 1.41 & $-$ & $-$ \\
$r_{10}$ & 10 & 4 & 344 347 & 2.45 & 8 374 082 & 4.59 & $-$ & $-$ \\
$r_6^2$ & 6 & 8 & 1 283 982 & 13.8 & 146 103 466 & 106. & 10.9 & $-$ \\
\hline
\end{tabular}
\label{table:bench2}
\end{center}
\end{table}
} 
{\begin{table}[!ht]
\begin{center}
\caption{$\multivsos$ vs $\raglib$ vs $\cad$ for positive polynomials over basic compact semialgebraic sets (Putinar).}
\begin{tabular}{lrr|rrr|c|c}
\hline
\multirow{2}{*}{Id} & \multirow{2}{*}{$n$} & \multirow{2}{*}{$d$}  & \multicolumn{3}{c|}{$\multivsos$} & $\raglib$ & $\cad$ \\
 & & & $k$ & $\tau_1$ (bits) & $t_1$ (s) & $t_2$ (s) & $t_3$ (s) \\
\hline  
$p_{46}$ & 2 & 4 & 3 & 21 723 & 0.83 & 0.15 & 0.81 \\
$f_{260}$ & 6 & 3 & 2 & 114 642 & 2.72 & 0.12 & $-$ \\
$f_{491}$ & 6 & 3 & 2 & 108 359& 9.65 & 0.01 & 0.05 \\
$f_{752}$ & 6 & 2 & 2 & 10 204 & 0.26 & 0.07 & $-$ \\
$f_{859}$ & 6 & 7 & 4 & 6 355 724 & 303. & 5896. & $-$ \\
$f_{863}$ & 4 & 2 & 1 & 5 492 & 0.14 & 0.01 & 0.01 \\
$f_{884}$ & 4 & 4 & 3 & 300 784 & 25.1 & 0.21 & $-$ \\
$f_{890}$ & 4 & 4 & 2 & 60 787 & 0.59 & 0.08 & $-$ \\
butcher & 6 & 3 & 2 & 247 623 & 1.32 & 47.2 & $-$ \\
heart & 8 & 4 & 2 &  618 847 & 2.94 & 0.54 & $-$ \\
magnetism & 7 & 2 & 1 & 9 622 & 0.29 & 434. & $-$ \\
\hline
\end{tabular}
\label{table:bench3}
\end{center}
\end{table}
}

Finally, we compare the performance of $\multivsos$ ($\putinarsos$) on
positive polynomials on basic compact semi-algebraic sets in
Table~\ref{table:bench3}. 
The first benchmark is
from~\cite[Problem~4.6]{Las01sos}. Each benchmark $f_i$ comes from an
inequality of the Flyspeck project~\cite{Hales_theflyspeck}. The three
last benchmarks are from~\cite{Munoz13}. The maximal degree of the
polynomials involved in each system is denoted by $d$. We emphasize
that the degree $D = 2k$ of each Putinar representation obtained in
practice with $\putinarsos$ is very close to $d$, which is in contrast
with the theoretical complexity estimates obtained in
Section~\ref{sec:putinar}. The values of $\varepsilon$, $\delta$ and
$\delta_c$ lie between $2^{-30}$ and $2^{-10}$, 60 and 200, 10 and 30.
As for Table~\ref{table:bench2}, $\raglib$ performs better for problems with $d \leq 3$ and $n \leq 4$. Larger problems (e.g.~magnetism, $f_{859}$) are handled more efficiently with $\multivsos$ and $\cad$ can only solve 3 benchmarks out of 10.
We plan to extend the procedure $\PP$ and the algorithm from~\cite{KLYZ08} to the case of such constrained problems.


\end{document}